\documentclass[11pt, a4paper, reqno]{article}
\usepackage{a4wide}
\usepackage{microtype}
\usepackage{amsthm}
\usepackage{amsmath}
\usepackage{amssymb}
\usepackage{hyperref}
\hypersetup{colorlinks=true,citecolor=blue,linkcolor=blue,urlcolor=blue}

\DeclareMathOperator{\VCSP}{VCSP}
\DeclareMathOperator{\MVCSP}{Max-VCSP}
\DeclareMathOperator{\MVCSPs}{Max-VCSP_s}
\DeclareMathOperator{\VCSPs}{VCSP_s}
\DeclareMathOperator{\Feas}{Feas}
\DeclareMathOperator{\Opt}{Opt}
\DeclareMathOperator{\Soft}{Soft}
\DeclareMathOperator{\mnrt}{Mn}
\DeclareMathOperator{\mjrt}{Mj}

\DeclareMathOperator{\ar}{ar}
\DeclareMathOperator{\opt}{opt}
\DeclareMathOperator{\sopt}{s-opt}
\DeclareMathOperator{\sub}{sub}

\newcommand{\mmorp}[1]{\langle #1 \rangle}

\newcommand{\tup}[1]{\ensuremath{\mathbf #1}}
\newcommand{\Q}{\mathbb{Q}}
\newcommand{\R}{\mathbb{R}}
\newcommand{\reduces}{\leq_{p}}
\newcommand{\Qn}{\ensuremath{\mathbb{Q}_{\geq 0}}}
\newcommand{\Qp}{\ensuremath{\mathbb{Q}_{> 0}}}
\newcommand{\QInfty}{\overline{\Q}}
\newcommand{\QnInfty}{\Qn \cup \{\infty\}}
\newcommand{\ignore}[1]{}
\newcommand{\closure}[1]{#1^*}
\newcommand{\clSingle}[1]{\closure{\{#1\}}}
\newcommand{\clGamma}{\closure{\Gamma}}
\newcommand{\gmcinst}{J}

\newcommand{\Tuv}{T_{\{u,v\}}}
\newcommand{\oTuv}{T_{uv}'}
\newcommand{\Round}[2]{\operatorname{Round}_{#1}(#2)}

\newcommand\Crestrict[2]{{% we make the whole thing an ordinary symbol
  \left.\kern-\nulldelimiterspace % automatically resize the bar with \right
  #1 % the function
  \right|_{#2} % this is the delimiter
  }}

\newtheorem{theorem}{Theorem}
\newtheorem{lemma}[theorem]{Lemma} 

\newtheorem{observation}[theorem]{Observation}
\newtheorem{corollary}[theorem]{Corollary}
\theoremstyle{definition}
\newtheorem{definition}[theorem]{Definition}
\newtheorem{example}[theorem]{Example}
\newtheorem{remark}[theorem]{Remark}

\begin{document}
\title{The complexity of Boolean surjective general-valued CSPs\thanks{Extended
abstracts of parts of this work appeared in \emph{Proceedings of the 18th
International Conference on Principles and Practice of Constraint Programming
(CP)}, pp. 38--54, 2012~\cite{Uppman12:cp} and in \emph{Proceedings of the 42nd
International Symposium on Mathematical Foundations of Computer Science
(MFCS)}~\cite{fz17:mfcs}.
Peter Fulla and Stanislav \v{Z}ivn\'y were
supported by a Royal Society Research Grant. Stanislav \v{Z}ivn\'y was supported
by a Royal Society University Research Fellowship. This project has received
funding from the European Research Council (ERC) under the European Union's
Horizon 2020 research and innovation programme (grant agreement No 714532). The
paper reflects only the authors' views and not the views of the ERC or the
European Commission. The European Union is not liable for any use that may be
made of the information contained therein.}}

\author{
Peter Fulla\\
University of Oxford, UK\\
\texttt{peter.fulla@cs.ox.ac.uk}
\and
Hannes Uppman\\
Link\"oping University, Sweden
\and
Stanislav \v{Z}ivn\'{y}\\
University of Oxford, UK\\
\texttt{standa.zivny@cs.ox.ac.uk}
}

\date{}
\maketitle

\begin{abstract}

Valued constraint satisfaction problems (VCSPs) are discrete optimisation
problems with a $(\Q\cup\{\infty\})$-valued objective function given as a sum of
fixed-arity functions.
In Boolean surjective VCSPs, variables take on labels from $D=\{0,1\}$ and an
optimal assignment is required to use both labels from $D$. Examples include the
classical \emph{global Min-Cut} problem in graphs and the \emph{Minimum
Distance} problem studied in coding theory.

We establish a dichotomy theorem and thus give a complete complexity
classification of Boolean surjective VCSPs with respect to exact solvability.
Our work generalises the dichotomy for $\{0,\infty\}$-valued constraint
languages (corresponding to surjective decision CSPs) obtained by Creignou and H\'ebrard.
For the maximisation problem of $\Qn$-valued surjective VCSPs, we also establish
a dichotomy theorem with respect to approximability.

Unlike in the case of Boolean surjective (decision) CSPs, there appears a novel tractable
class of languages that is trivial in the non-surjective setting. This newly
discovered tractable class has an interesting mathematical structure related to
downsets and upsets. Our main contribution is identifying this class and proving
that it lies on the borderline of tractability. A crucial part of our proof is a
polynomial-time algorithm for enumerating all near-optimal solutions to a
generalised Min-Cut problem, which might be of independent interest.
\end{abstract}

\section{Introduction}
\label{sec:intro}

The framework of valued constraint satisfaction problems (VCSPs) captures many
fundamental discrete optimisation problems. A VCSP \emph{instance} $I =
(V,D,\phi_I)$ is given by a finite set of variables $V = \{ x_1, \dots, x_n \}$,
a finite set of labels $D$ called the \emph{domain}, and an objective function
$\phi_I : D^n \to \QInfty$, where $\QInfty = \Q \cup \{\infty\}$ denotes the set
of extended rationals. The objective function is expressed by a weighted sum of
\emph{valued constraints}
\begin{equation}
\phi_I(x_1, \dots, x_n) = \sum_{i=1}^q w_i \cdot \gamma_i(\tup{x}_i) \,,
\end{equation}
where $\gamma_i : D^{\ar(\gamma_i)} \to \QInfty$ is a \emph{weighted relation}
of \emph{arity} $\ar(\gamma_i) \in \mathbb{Z}_{\geq 1}$, $w_i \in \Qn$ is the
\emph{weight} and $\tup{x}_i \in V^{\ar(\gamma_i)}$ the \emph{scope} of the
$i$th valued constraint. (Note that zero weights are allowed; we define $0 \cdot
\infty = \infty$.) The value of an assignment of domain labels to variables $s :
V \to D$ equals $\phi_I(s) = \phi_I(s(x_1), \dots, s(x_n))$. An assignment $s$
is \emph{feasible} if $\phi_I(s) < \infty$, and it is \emph{optimal} if it is
feasible and $\phi_I(s) \leq \phi_I(s')$ for all assignments $s'$. Given an
instance $I$, the goal is to find an optimal assignment, i.e. one that minimises
$\phi_I$. A \emph{valued constraint language} (or just a language) $\Gamma$ is
set of weighted relations over a domain $D$. We denote by $\VCSP(\Gamma)$ the
class of all VCSP instances that use only weighted relations from a language
$\Gamma$ in their objective function. VCSPs are also called general-valued
CSPs~\cite{Kolmogorov17:sicomp} to emphasise the fact that (decision) CSPs are a
special case of VCSPs in which weighted relations only assign values $0$ and
$\infty$. (However, $\Q$-valued VCSPs~\cite{tz16:jacm} do \emph{not} include
CSPs as a special case.)

For an example of a VCSP, consider the $(s,t)$-Min-Cut
problem~\cite{Schrijver03:CombOpt}. Given a digraph $G=(V,E)$ with a source
$s\in V$, sink $t\in V$, and edge weights $w:E\to\Qp$, the goal is to find a set
$C\subseteq V$ with $s\in C$ and $t\not\in C$ that minimises
\begin{equation}
\sum_{(u,v)\in E,u\in C,v\not\in C} w(u,v) \,.
\end{equation}
We show how the $(s,t)$-Min-Cut problem can be expressed as a VCSP over a domain
$D = \{ 0, 1 \}$ (a domain of size $2$ such as this one is called
\emph{Boolean}). We define a language $\Gamma_{\sf cut} = \{ \rho_0, \rho_1,
\gamma \}$ as follows: For $d \in D$, $\rho_d : D \to \QInfty$ is defined by
$\rho_d(x)=0$ if $x=d$ and $\rho_d(x)=\infty$ if $x\neq d$. Weighted relation
$\gamma : D^2 \to \QInfty$ is defined by $\gamma(x,y)=1$ if $x=0$ and $y=1$, and
$\gamma(x,y)=0$ otherwise. Given an $(s,t)$-Min-Cut instance on a digraph
$G=(V,E)$, the problem of finding an optimal $(s,t)$-Min-Cut in $G$ is
equivalent to solving an instance $I=(V,D,\phi_I)$ of $\VCSP(\Gamma_{\sf cut})$
such that
\begin{equation}
\phi_I(x_1, \dots, x_n) =
\rho_0(s) + \rho_1(t) + \sum_{(u,v) \in E} w(u,v) \cdot \gamma(u, v) \,.
\end{equation}
It is well known that the $(s,t)$-Min-Cut problem is solvable in polynomial
time. Since every instance $I$ of $\VCSP(\Gamma_{\sf cut})$ can be reduced to an
instance of the $(s,t)$-Min-Cut problem, $\VCSP(\Gamma_{\sf cut})$ is solvable
in polynomial time.

A language $\Gamma$ is called \emph{tractable} if, for every finite $\Gamma'
\subseteq \Gamma$, $\VCSP(\Gamma')$ is solvable in polynomial time. If there
exists a finite $\Gamma' \subseteq \Gamma$ such that $\VCSP(\Gamma')$ is
NP-hard, then $\Gamma$ is called \emph{intractable}.\footnote{Defining
tractability in terms of finite subsets ensures that the tractability of a
language is independent of whether the weighted relations are represented
explicitly (by tables of values) or implicitly (by oracles).} For example,
language $\Gamma_{\sf cut}$ is tractable. It is natural to ask about the
complexity of $\VCSP(\Gamma)$ for a fixed language $\Gamma$. Cohen et
al.~\cite{cohen06:complexitysoft} obtained a dichotomy classification of
\emph{Boolean} languages: They identified eight tractable classes (one of which
correspons to submodularity~\cite{Schrijver03:CombOpt} and includes $\Gamma_{\sf
cut}$) and showed that the remaining languages are intractable. The dichotomy
classification from~\cite{cohen06:complexitysoft} is an extension of Schaefer's
celebrated result~\cite{Schaefer78:complexity}, which gave a dichotomy for
Boolean $\{0, \infty\}$-valued constraint languages, and the work of
Creignou~\cite{Creignou95:jcss}, which established a dichotomy classification
for Boolean $\{0,1\}$-valued constraint languages.

The \emph{surjective} variant of VCSPs further requires that assignments of
domain labels to variables be surjective (an assignment $s : V \to D$ is
surjective if, for every $d \in D$, there exists $x \in V$ such that $s(x) =
d$). Thus, the goal is to find an assignment that is optimal among surjective
assignments. For Boolean VCSPs with $D=\{0,1\}$, this simply means that the
all-zero and all-one assignments are disregarded. We define $\VCSP(\Gamma)$,
tractability, and intractability in the surjective setting analogously with
regular VCSPs, and refer to them as $\VCSPs(\Gamma)$, s-tractability, and
s-intractability.

For an example of a surjective VCSP, consider the (global) Min-Cut
problem~\cite{Schrijver03:CombOpt}. Given a graph $G=(V,E)$ and edge weights
$w:E\to\Qp$, the goal is to find a set $C\subseteq V$ with $\emptyset \subsetneq
C \subsetneq V$ that minimises
\begin{equation}
\sum_{\{u,v\}\in E,|\{u,v\}\cap C|=1}w(u,v) \,.
\end{equation}
Again, this problem can be expressed over a Boolean domain $D = \{ 0, 1 \}$. We
define a weighted relation $\gamma : D^2 \to \QInfty$ by $\gamma(x,y)=0$ if
$x=y$ and $\gamma(x,y)=1$ if $x\neq y$. Then the problem of finding an optimal
Min-Cut in a graph $G=(V,E)$ is equivalent to solving an instance $I =
(V,D,\phi_I)$ of $\VCSPs(\{ \gamma \})$ such that
\begin{equation}
\phi_I(x_1, \dots, x_n) = \sum_{\{u,v\} \in E} w(u,v) \cdot \gamma(u, v) \,.
\end{equation}
Note that the two non-surjective assignments to $I$ correspond to sets
$\emptyset$ and $V$, which are not admissible solutions to the Min-Cut problem.
Since every instance of $\VCSPs(\{ \gamma \})$ can be straightforwardly
translated to a Min-Cut instance, and the Min-Cut problem is solvable in
polynomial time (say, by a reduction to the $(s,t)$-Min-Cut problem, though
other algorithms exist~\cite{Soter97:jacm}), language $\{ \gamma \}$ is
s-tractable.

The computational complexity of $\VCSP(\Gamma)$ and $\VCSPs(\Gamma)$ is closely
related. Namely, $\VCSP(\Gamma)$ is polynomial-time reducible to
$\VCSPs(\Gamma)$ (see Lemma~\ref{lmReductionVCSPtoSurjective}), i.e., any intractable
language is also s-intractable. Let $\mathcal{C}_D = \{ \rho_d ~|~ d \in D \}$,
where we define $\rho_d : D \to \QInfty$ by $\rho_d(x) = 0$ if $x = d$ and
$\rho_d(x) = \infty$ if $x \neq d$; these unary weighted relations are called
\emph{constants}. Conversely, $\VCSPs(\Gamma)$ is polynomial-time reducible to
$\VCSP(\Gamma \cup \mathcal{C}_D)$ (see Lemma~\ref{lmReductionSurjectiveToVCSP}),
i.e., any tractable language containing constants $\mathcal{C}_D$ is also
s-tractable. In the case of Boolean $\{ 0, \infty \}$-valued languages,
Schaefer's dichotomy involves six tractable classes. Four of them include
constants, and hence they are s-tractable. Creignou and
H\'ebrard~\cite{Creignou97} showed that the remaining two classes ($0$-valid and
$1$-valid\footnote{A $\{ 0, \infty \}$-valued weighted relation is $0$-valid
($1$-valid) if it assigns value $0$ to the all-zero (all-one) tuple.}) are
s-intractable, thus obtaining a dichotomy classification of Boolean $\{ 0,
\infty \}$-valued languages in the surjective setting.

\subsection*{Contributions}

\paragraph{Complexity classification}
As our main contribution, we establish a dichotomy classification of all Boolean
($\QInfty$-valued) languages in the surjective setting, which extends the
classification from~\cite{Creignou97}. Let $D = \{ 0, 1 \}$. Six of the eight
tractable classes of Boolean languages identified by Cohen et
al.~\cite{cohen06:complexitysoft} include constants $\mathcal{C}_D$, and thus
are also s-tractable. We show that languages in the remaining two classes
($0$-optimal and $1$-optimal\footnote{A weighted relation is $0$-optimal
($1$-optimal) if the all-zero (all-one) tuple minimises it.}) are s-tractable
if, for every weighted relation, the set of feasible tuples and the set of
optimal tuples are \emph{essentially downsets} (in the $0$-optimal case; see
Definition~\ref{defEssentiallyDownset}) or \emph{essentially upsets} (in the $1$-optimal
case); otherwise, they are s-intractable.

Somewhat surprisingly, such languages are s-tractable regardless of the
remaining (i.e., finite but non-optimal) values. Those values must, however,
bear on the time bound of any polynomial-time algorithm solving surjective VCSPs
over such languages (unless P = NP). In particular, we give an example of an
infinite language $\Gamma$ that is s-tractable (i.e., $\VCSPs(\Gamma')$ can be
solved in polynomial time for every finite $\Gamma' \subseteq \Gamma$) but
$\VCSPs(\Gamma)$ is NP-hard (see Example~\ref{exGloballySIntractable}). This is quite
unusual; all known tractable classes of VCSPs are in fact \emph{globally
tractable}, which means that $\VCSP(\Gamma')$ is solvable by the same
polynomial-time algorithm for every finite subset $\Gamma'$ of a tractable
language $\Gamma$, and hence $\VCSP(\Gamma)$ is also polynomial-time
solvable~\cite{Bulatov05:classifying}. To capture this distinction, our main
result (Theorem~\ref{thmClassificationSurjectiveVCSP}) gives a classification in terms
of \emph{global} s-tractability,\footnote{Weighted relations in an instance are
assumed to be represented explicitly (by tables of values). We only consider
languages of bounded arity; this restriction is vital in some of our proofs.
Also, unbounded arity presents new challenges to complexity classification. For
example, explicitly representing a weighted relation of an arity that is
super-logarithmic in the number of variables requires super-polynomial space.}
from which a classification for s-tractability easily follows (see
Remark~\ref{remLocalClsSrjectiveVCSP}). We call the condition that describes the
borderline of global s-tractability in the $0$-optimal case \emph{EDS} (see
Definition~\ref{defWRelEDS}), drawing a parallel to the corresponding condition
for s-tractability, which involves essentially downsets. The $1$-optimal case is
analogous (one only needs to exchange the roles of labels $0$ and $1$).

\paragraph{Tractability}
While $0$-optimal and $1$-optimal languages are trivially tractable for VCSPs,
the algorithm for surjective VCSPs over the newly identified class of languages
is nontrivial and constitutes our second main contribution.
The global s-tractability part of our result is established by a reduction from
$\QInfty$-valued $\VCSPs$ to the \emph{generalised Min-Cut} problem (defined in
Section~\ref{secSurjectiveTractability}), for which we require to find all
$\alpha$-optimal solutions in polynomial time, where $\alpha$ is a constant
depending on the valued constraint language.
The generalised Min-Cut problem consists in minimising an objective function
$f+g$, where $f$ is a superadditive set function given by an oracle and $g$ is a
cut function (same as in the Min-Cut problem);
see~Section~\ref{secSurjectiveTractability} for the details.
We prove that the running time of our algorithm is roughly
$O\left(n^{20\alpha}\right)$, thus improving on the bound of
$O\left(n^{3^{3\alpha}}\right)$ established in~\cite{Uppman12:cp} (one of the
two extended conference abstracts of this paper) for the special case of
$\{0,1\}$-valued languages.

\paragraph{Hardness}
The hardness part of our result is proved by analysing weighted relations that
can be obtained from a language using gadgets that preserve (global)
s-tractability. Since not all standard gadgets have this property (in
particular, minimisation over a variable may affect the surjectivity of a
solution), we cannot employ the algebraic approach~\cite{cccjz13:sicomp}.
Instead, we define a collection of operations that form building blocks of
gadgets preserving tractability in the surjective setting (see
Definition~\ref{defSurjectiveClosure} and
Lemma~\ref{lmSurjectiveClosureReduction}). Such gadgets apply to non-Boolean
domains as well, and may be useful in future work on non-Boolean surjective
VCSPs. Another important ingredient of our proof is the NP-hardness of the
\emph{Minimum Distance} problem~\cite{Vardy1997}, which to the best of our
knowledge has not previously appeared in the literature on exact solvability of
(V)CSPs.

\paragraph{Approximability}
By a simple reduction, our main result implies a complexity classification of
the approximability of maximising $\Qn$-valued surjective VCSPs (see
Theorem~\ref{thmApproximabilitySurjectiveFiniteValuedVCSP}).

\paragraph{Enumeration}
For the globally s-tractable languages, we also show that \emph{all} optimal
solutions can be enumerated with polynomial
delay~\cite{Valiant79:sicomp-complexity} (see
Theorem~\ref{thmEnumerationSurjectiveVCSP}). While this is an easy observation
for the already known globally s-tractable languages (since constants
$\mathcal{C}_D$ allow for a standard self-reduction technique), we prove the
same result for the newly discovered classes of languages, which do \emph{not}
include constants $\mathcal{C}_D$.

\subsection*{Related work}
Recent years have seen some remarkable progress on the computational complexity
of CSPs and VCSPs parametrised by the (valued) constraint language. We highlight
the resolution of the ``bounded width conjecture''~\cite{Barto14:jacm} and the
result that a dichotomy for CSPs, conjectured in~\cite{Feder98:monotone} and
recently established by two independent
proofs~\cite{Bulatov17:focs,Zhuk17:focs}, implies a dichotomy for
VCSPs~\cite{Kozik15:icalp,Kolmogorov17:sicomp}. All this work is for arbitrary
(i.e., not necessarily Boolean) finite domains and relies on the algebraic
approach initiated in~\cite{Bulatov05:classifying} and nicely described in a
survey~\cite{Barto17:survey}.

One of the important aspects of the algebraic approach is the assumption that
constants $\mathcal{C}_D$ are present in (valued) constraint languages. (This is
without loss of generality with respect to polynomial-time solvability.) In the
surjective setting, it is the lack of constants that makes it difficult, if not
impossible, to employ the algebraic approach. Chen made the first step in this
direction~\cite{Chen2014} but it is not clear how to take his result (for CSPs)
further.

For a binary (unweighted) relation $\gamma$, $\VCSPs(\{\gamma\})$ has been
studied under the name of surjective
$\gamma$-Colouring~\cite{Bodirsky12:dam,Golovach12:tcs,Golovach12:acta,Martin15:jctb}
and vertex-compaction~\cite{Vikas13}. We remark that our notion of surjectivity
is global. For the $\gamma$-Colouring problem, a local version of surjectivity
has also been studied~\cite{Fiala05:tcs,Fiala08:csr}. This version corresponds
to finding a graph homomorphism such that the neighbourhood of every vertex $v$
is mapped surjectively onto the neighbourhood of the image of $v$.

Under the assumption of the unique games conjecture~\cite{khot10:coco},
Raghavendra has shown that the optimal approximation ratio for maximising
$\Qn$-valued VCSPs is achieved by the basic semidefinite programming
relaxation~\cite{raghavendra08:stoc,Raghavendra}.

Bach and Zhou have shown that any Max-CSP that is solvable in polynomial time in
the non-surjective setting admits a PTAS in the surjective setting, and that any
Max-CSP that is APX-hard in the non-surjective setting remains APX-hard in the
surjective setting~\cite{Bach11:surj}.

\section{Preliminaries}

\subsection{Weighted relations and VCSPs}
\label{ssecWeightedRelations}

We work in the arithmetic model of computation, i.e., every number is
represented in constant space, and basic arithmetic operations take constant
time. Let $\QInfty = \Q \cup \{\infty\}$ denote the set of extended rationals.
For any $c \in \QInfty$, we define $c \leq \infty$ and $\infty + c = c + \infty
= \infty$. If $c \geq 0$, we define $c \cdot \infty = \infty \cdot c = \infty$.
We leave the result of multiplying $\infty$ undefined for $c < 0$.

For any integer $n \geq 1$, let $[n] = \{1, \dots, n\}$.

\begin{definition}
\label{defWeightedRelation}

Let $r \geq 1$ be an integer. An $r$-ary \emph{weighted relation} over $D$ is a
mapping $\gamma : D^r \to \QInfty$; the \emph{arity} of $\gamma$ equals
$\ar(\gamma) = r$. We denote by $\Feas(\gamma)$ the underlying \emph{feasibility
relation} of $\gamma$, i.e.
\begin{equation}
\Feas(\gamma) =
\left\{
\tup{x} \in D^r
~\middle|~
\gamma(\tup{x}) < \infty
\right\} \,.
\end{equation}
We denote by $\Opt(\gamma)$ the relation consisting of the minimal-valued
tuples, i.e.
\begin{equation}
\Opt(\gamma) =
\left\{
\tup{x} \in \Feas(\gamma)
~\middle|~
\gamma(\tup{x}) \leq \gamma(\tup{y}) \text{ for every } \tup{y} \in D^r
\right\} \,.
\end{equation}

A weighted relation $\gamma$ is called \emph{crisp} if $\Feas(\gamma) =
\Opt(\gamma)$. In other words, there exists a constant $c \in \Q$ such that
$\gamma(\tup{x}) = c$ for all $\tup{x} \in \Feas(\gamma)$ and $\gamma(\tup{x}) =
\infty$ for all $\tup{x} \in D^r \setminus \Feas(\gamma)$.

\end{definition}

Weighted relations that differ only by a constant are considered equivalent, as
adding a rational constant to a weighted relation changes the value of every
solution to the $\VCSP$ by the same amount. Therefore, a crisp weighted relation
$\gamma$ can be equated with the relation $\Feas(\gamma)$. Conversely, a
relation $\rho$ can be seen as a crisp weighted relation $\gamma_c$ with
$\Feas(\gamma_c) = \rho$ and the codomain equal to $\{c, \infty\}$ for some $c
\in \Q$. Unless stated otherwise, we choose $c = 0$.

\begin{definition}
\label{defCommonRelations}

We denote by $\rho_=$ the binary equality relation $\{ (d,d) ~|~ d \in D \}$.
For any $d \in D$, we denote by $\rho_d$ the unary relation $\{ (d) \}$, which
is called a \emph{constant}. We denote the set of constants on $D$ by
$\mathcal{C}_D = \{ \rho_d ~|~ d \in D \}$.

For any relation $\rho$, we denote by $\Soft(\rho)$ the \emph{soft variant} of
$\rho$ defined by $\Soft(\rho)(\tup{x}) = 0$ if $\tup{x} \in \rho$ and
$\Soft(\rho)(\tup{x}) = 1$ otherwise.

\end{definition}

\begin{definition}

A \emph{constraint language} (or simply a \emph{language}) over $D$ is a
(possibly infinite) set of weighted relations over $D$.

\end{definition}

In this paper, we only consider languages of bounded arity. Note that a crisp
language is of a bounded arity if and only if it is finite.

\begin{definition}

A language $\Gamma$ is called \emph{s-tractable} if, for every finite $\Gamma'
\subseteq \Gamma$, $\VCSPs(\Gamma')$ can be solved in polynomial time. If
$\VCSPs(\Gamma)$ can be solved in polynomial time, language $\Gamma$ is called
\emph{globally s-tractable}.

If there exists a finite $\Gamma' \subseteq \Gamma$ such that $\VCSPs(\Gamma')$
is NP-hard, language $\Gamma$ is called \emph{s-intractable}. If
$\VCSPs(\Gamma)$ is NP-hard, language $\Gamma$ is called \emph{globally
s-intractable}.

\end{definition}

Note that a globally s-tractable language is s-tractable, and an s-intractable
language is globally s-intractable.

Lemmas~\ref{lmReductionVCSPtoSurjective}~and~\ref{lmReductionSurjectiveToVCSP} establish a
relation between the complexity of the $\VCSP$ and $\VCSPs$. We denote by
$\reduces$ the standard polynomial-time Turing reduction.

\begin{lemma}
\label{lmReductionVCSPtoSurjective}

For any constraint language $\Gamma$,
\begin{equation}
\VCSP(\Gamma) \reduces \VCSPs(\Gamma) \,.
\end{equation}

\end{lemma}

\begin{proof}

Given an instance $I$ of $\VCSP(\Gamma)$, we construct an instance $I'$ of
$\VCSPs(\Gamma)$ by adding $|D|$ extra variables. Any solution to $I$ can be
extended to a surjective solution to $I'$ of the same value and, conversely, any
(surjective) solution to $I'$ induces a solution to $I$ of the same value.
\qedhere

\end{proof}

\begin{lemma}
\label{lmReductionSurjectiveToVCSP}

For any constraint language $\Gamma$,
\begin{equation}
\VCSPs(\Gamma) \reduces \VCSP(\Gamma \cup \mathcal{C}_D) \,.
\end{equation}

\end{lemma}

\begin{proof}

Given an instance $I = (V, D, \phi_I)$ of $\VCSPs(\Gamma)$, we iterate through
all $O \left( |V|^{|D|} \right)$ injective mappings $f: D \to V$. For each
mapping $f$, we construct an instance $I'_f$ of $\VCSP(\Gamma \cup
\mathcal{C}_D)$ by adding constraints $\rho_d(f(d))$ for all $d \in D$. The
additional constraints guarantee that only surjective solutions to $I'_f$ are
feasible. Conversely, any surjective solution to $I$ is a feasible solution to
$I'_f$ for some mapping $f$. Therefore, a solution of the smallest value among
optimal solutions to $I'_f$ for all $f$ is an optimal surjective solution to
$I$. \qedhere

\end{proof}

\begin{corollary}

Any (globally) tractable language $\Gamma$ with $\mathcal{C}_D \subseteq \Gamma$
is also (globally) s-tractable.

\end{corollary}

Now we define a few operations on weighted relations that occur throughout the
paper.

\begin{definition}
\label{defWRelOperations}

Let $\gamma$ be an $r$-ary weighted relation.

\begin{itemize}

\item \emph{Addition of a rational constant}: For any $c \in \Q$, $\gamma + c =
\gamma'$ such that $\gamma'(\tup{x}) = \gamma(\tup{x}) + c$.

\item \emph{Non-negative scaling}: For any $c \in \Qn$, $c \cdot \gamma =
\gamma'$ such that $\gamma'(\tup{x}) = c \cdot \gamma(\tup{x})$. Note that $0
\cdot \gamma = \Feas(\gamma)$.

\item \emph{Coordinate mapping}: For any arity $r'$ and mapping $f : [r] \to
[r']$, $f(\gamma) = \gamma'$ such that $\gamma'(x_1, \dots, x_{r'}) = \gamma
\left( x_{f(1)}, \dots, x_{f(r)} \right)$.

\item \emph{Minimisation}: For any $i \in [r]$, the minimisation of $\gamma$ at
coordinate $i$ results in $\gamma'$ such that $\gamma'(x_1, \dots, x_{i-1},
x_{i+1}, \dots, x_r) = \min_{x_i \in D} \gamma(x_1, \dots, x_r)$.

\item \emph{Pinning}: For any $d \in D$ and $i \in [r]$, the pinning of $\gamma$
to label $d$ at coordinate $i$ results in $\gamma'$ such that $\gamma'(x_1,
\dots, x_{i-1}, x_{i+1}, \dots, x_r) = \gamma(x_1, \dots, x_{i-1}, d, x_{i+1},
\dots, x_r)$. A pinning to label $d$ is called a \emph{$d$-pinning}.

\item \emph{Addition}: For any weighted relations $\gamma_1, \gamma_2$ with
$\ar(\gamma_1) = \ar(\gamma_2)$, $\gamma_1 + \gamma_2 = \gamma'$ such that
$\gamma'(\tup{x}) = \gamma_1(\tup{x}) + \gamma_2(\tup{x})$.

\end{itemize}

\end{definition}

We extend operations on weighted relations to languages in the natural way,
e.g., $\Feas(\Gamma) = \{ \Feas(\gamma) ~|~ \gamma \in \Gamma \}$.

A weighted relational clone~\cite{cccjz13:sicomp} is a language closed under
certain operations (e.g., non-negative scaling and minimisation) that preserve
the tractability of languages in the following sense: The VCSP over the smallest
weighted relational clone containing a language $\Gamma$ can be reduced in
polynomial time to $\VCSP(\Gamma)$. Weighted relational clones are characterised
by their weighted polymorphisms (a generalisation of multimorphisms defined in
Definition~\ref{def:mult}), which enables the employment of tools from universal algebra
in the effort to obtain a complexity classification of languages.

In the surjective setting, however, minimisation may not preserve the
tractability of languages, and thus we need to define a language closure that
excludes this operation. Consequently, we are unable to use the algebraic
approach in our proofs in Section~\ref{secSurjectiveHardness}.

\begin{definition}
\label{defSurjectiveClosure}

A constraint language $\Gamma$ is called \emph{closed} if it is closed under
addition, coordinate mapping, non-negative scaling, addition of a rational
constant, operation $\Opt$, and, for all $d \in D$ such that $\rho_d\in\Gamma$,
$d$-pinning.

We define $\clGamma$ to be the smallest closed language containing $\Gamma$.

\end{definition}

Now we show that these closure operations preserve the complexity of the
$\VCSPs$. Note that we require a language to be closed under $d$-pinning
\emph{only} if it contains $\rho_d$.

\begin{lemma}
\label{lmSurjectiveClosureReduction}

For any constraint language $\Gamma$,
\begin{equation}
\VCSPs(\clGamma)\reduces\VCSPs(\Gamma) \,.
\end{equation}

\end{lemma}

\begin{proof}

For most of the closure operations, standard reductions for the $\VCSP$ apply to
the surjective setting as well. Let $\gamma_1, \gamma_2 \in \Gamma$ be weighted
relations with $\ar(\gamma_1) = \ar(\gamma_2)$, and let $\gamma' = \gamma_1 +
\gamma_2$. Then $\VCSPs(\Gamma \cup \{ \gamma' \}) \reduces \VCSPs(\Gamma)$, as
any constraint of the form $w \cdot \gamma'(\tup{x})$ can be replaced with a
pair of constraints $w \cdot \gamma_1(\tup{x})$, $w \cdot \gamma_2(\tup{x})$.
Similarly, let $\gamma \in \Gamma$ and $\gamma' = f(\gamma)$ where $f :
[\ar(\gamma)] \to [\ar(\gamma')]$; then any constraint of the form $w \cdot
\gamma' \left( x_1, \dots, x_{\ar(\gamma')} \right)$ can be replaced with a
constraint $w \cdot \gamma \left( x_{f(1)}, \dots, x_{f(\ar(\gamma))} \right)$.
Non-negative scaling can be achieved by scaling the weight of affected
constraints. Addition of a rational constant changes the value of every solution
by the same amount, and thus it can be ignored.

Now we show that $\VCSPs(\Gamma \cup \{ \Opt(\gamma) \}) \reduces
\VCSPs(\Gamma)$ for any $\gamma \in \Gamma$. Let $I$ be an instance of
$\VCSPs(\Gamma \cup \{ \Opt(\gamma) \})$. Without loss of generality, assume
that the minimum values assigned by $\gamma$ and $\Opt(\gamma)$ equal $0$ and
all weighted relations in $I$ assign non-negative values (this can be achieved
by adding rational constants). We may also assume that $\gamma$ is not crisp (otherwise
$\Opt(\gamma) = \gamma$). Let $m$ denote the smallest positive value assigned by
$\gamma$, and let $M$ be an upper bound on the value of any feasible solution to
$I$ (e.g., the weighted sum of the maximum finite values assigned by the
constraints of $I$). We replace every constraint of the form $w \cdot
\Opt(\gamma)(\tup{x})$ in $I$ with a constraint $(M/m + 1) \cdot
\gamma(\tup{x})$ to obtain an instance $I' \in \VCSPs(\Gamma)$. Any feasible
solution to instance $I$ gets assigned the same value by $I'$. Any infeasible
solution to instance $I$ is either infeasible for $I'$ as well, or it incurs an
infinite value from a constraint of the form $w \cdot \Opt(\gamma)(\tup{x})$ in
$I$ and thus a value of at least $(M/m + 1) \cdot m > M$ in $I'$. Therefore, an
optimal solution to $I'$ is optimal for $I$ as well.

In the case of pinning, we need a different reduction as the standard one relies
on minimisation. Suppose that $\rho_d \in \Gamma$. Let $\gamma'$ be a
$d$-pinning of a weighted relation $\gamma \in \Gamma$; without loss of
generality, let it be a pinning at the first coordinate. We show that
$\VCSPs(\Gamma \cup \{\gamma'\}) \reduces \VCSPs(\Gamma)$. Let $I =
(V,D,\phi_I)$ be an instance of $\VCSPs(\Gamma\cup\{\gamma'\})$ with $V = \{x_1,
\dots, x_n\}$. In a surjective solution to $I$, at least one variable is
assigned label $d$, but we do not a priori know which one. For every $i \in
[n]$, we construct an instance $I_i=(V,D,\phi_{I_i})$ of $\VCSPs(\Gamma)$ by
replacing all constraints of the form $\gamma'(\tup{x})$ with $\gamma(x_i,
\tup{x})$ and adding a constraint $\rho_d(x_i)$ to force variable $x_i$ to take
label $d$. A solution of the smallest value among optimal solutions to $I_1,
\dots, I_n$ is an optimal solution to $I$. \qedhere

\end{proof}

\subsection{Polymorphisms and multimorphisms}

For any $r \geq 1$ and a $k$-ary operation $h: D^k \to D$, we extend $h$ to
$r$-tuples over $D$ by applying it componentwise. Namely, for $\tup{x}_1, \dots,
\tup{x}_k \in D^r$ where $\tup{x}_i = (x_{i,1}, \dots, x_{i,r})$, we define
$h(\tup{x}_1, \dots, \tup{x}_k) \in D^r$ by
\begin{equation}
h \left( \tup{x}_1, \dots, \tup{x}_k \right)
=
\left(
h \left( x_{1,1}, \dots, x_{k,1} \right), \dots,
h \left( x_{1,r}, \dots, x_{k,r} \right)
\right) \,.
\end{equation}

The following notion is at the heart of the algebraic approach to decision
CSPs~\cite{Bulatov05:classifying}.

\begin{definition}
Let $\gamma$ be a weighted relation on $D$.
A $k$-ary operation
$h:D^k\to D$ is a \emph{polymorphism} of $\gamma$ (and $\gamma$ is
\emph{invariant under} or \emph{admits} $h$) if, for every
$\tup{x}_1,\ldots,\tup{x}_k\in\Feas(\gamma)$, we have
$h(\tup{x}_1,\ldots,\tup{x}_k)\in\Feas(\gamma)$. We say that $h$ is a
polymorphism of a language $\Gamma$ if it is a polymorphism
of every $\gamma\in\Gamma$.
\end{definition}

The following notion, which involves a collection of $k$ $k$-ary polymorphisms,
plays an important role in the complexity classification of Boolean valued constraint
languages~\cite{cohen06:complexitysoft}, as we will see in Theorem~\ref{thmClassificationVCSP} in Section~\ref{sec:Boolean}.

\begin{definition}\label{def:mult}
Let $\gamma$ be a weighted relation on $D$.
A list $\mmorp{h_1,\ldots,h_k}$ of $k$-ary polymorphisms of
$\gamma$ is a $k$-ary \emph{multimorphism} of $\gamma$ (and $\gamma$ \emph{admits}
$\mmorp{h_1,\ldots,h_k}$) if, for every
$\tup{x}_1,\ldots,\tup{x}_k\in\Feas(\gamma)$, we have
\begin{equation*}
\sum_{i=1}^k \gamma(h_i(\tup{x}_1,\ldots,\tup{x}_k))\ \leq\ \sum_{i=1}^k
\gamma(\tup{x}_i)\,.
\end{equation*}
$\mmorp{h_1,\ldots,h_k}$ is a multimorphism of a language $\Gamma$ if it is a multimorphism of every $\gamma\in\Gamma$.
\end{definition}

The operations in Definition~\ref{defSurjectiveClosure} preserve
multimorphisms~\cite{cohen06:complexitysoft,fz16:toct}, i.e., any multimorphism
of a language $\Gamma$ is also a multimorphism of $\clGamma$. Consequently, all
polymorphisms of a crisp weighted relation are preserved.

\subsection{Boolean VCSPs}
\label{sec:Boolean}

In the rest of the paper, we consider only Boolean languages (i.e., $D = \{0,
1\}$), unless explicitly mentioned otherwise. For any arity $r \geq 1$, we
denote by $\tup{0}^r$ ($\tup{1}^r$) the zero (one) $r$-tuple. For $r$-tuples
$\tup{x} = (x_1, \dots, x_r)$ and $\tup{y} = (y_1, \dots, y_r) \in D^r$, we
define $\tup{x} \leq \tup{y}$ if and only if $x_i \leq y_i$ for all $i \in [r]$
(where $0 < 1$). We also define the following operations on $D$:

\begin{itemize}
\item
For any $a \in D$,
$c_a$ is the constant unary operation such that $c_a(x) = a$ for all $x \in D$.
\item
Operation $\neg$ is the unary negation, i.e.\ $\neg(0) = 1$ and $\neg(1) = 0$.
For a weighted relation $\gamma$, we define $\neg(\gamma)$ to be the weighted
relation $\neg(\gamma)(\tup{x})=\gamma(\neg(\tup{x}))$. For a language $\Gamma$,
we define $\neg(\Gamma)=\{\neg(\gamma)~|~\gamma\in\Gamma\}$. Note that
$\neg(\Gamma)$ can be obtained from $\Gamma$ simply by exchanging the labels $\{
0, 1 \}$, and hence has the same complexity as $\Gamma$.
\item
Binary operation $\oplus$ is the addition modulo $2$ operation. In this case, we
use the infix notation, i.e., $0 \oplus 0 = 0 = 1 \oplus 1$ and $0 \oplus 1 = 1
= 1 \oplus 0$.
\item
Binary operation $\min$ ($\max$) returns the smaller (larger) of its two
arguments with respect to the order $0 < 1$.
\item
Binary operation $\sub$ (for subtraction) is defined as $\sub(x, y) = \min(x,
\neg y)$.
\item
Ternary operation $\mnrt$ (for minority) is the unique ternary operation on $D$
satisfying $\mnrt(x,x,y)=\mnrt(x,y,x)=\mnrt(y,x,x)=y$ for all $x,y\in D$.
\item
Ternary operation $\mjrt$ (for majority) is the unique ternary operation on $D$
satisfying $\mjrt(x,x,y)=\mjrt(x,y,x)=\mjrt(y,x,x)=x$ for all $x,y\in D$.
\end{itemize}

\begin{lemma}
\label{lmSubImpliesC0Min}

If a weighted relation admits polymorphism $\sub$, then it also admits
polymorphisms $c_0$ and $\min$.

\end{lemma}

\begin{proof}

For every $x, y \in D$, it holds $c_0(x) = 0 = \sub(x, x)$ and $\min(x, y) =
\sub(x, \sub(x, y))$. \qedhere

\end{proof}

Cohen et al.~\cite{cohen06:complexitysoft} established a complexity
classification of Boolean constraint languages.

\begin{theorem}[{\cite[Theorem 7.1]{cohen06:complexitysoft}}]
\label{thmClassificationVCSP}
Let $\Gamma$ be a Boolean $\QInfty$-valued language. Then $\Gamma$ is tractable
if it admits any the following multimorphisms: $\mmorp{c_0}$, $\mmorp{c_1}$,
$\mmorp{\min, \min}$, $\mmorp{\max, \max}$, $\mmorp{\min, \max}$, $\mmorp{\mnrt,
\mnrt, \mnrt}$, $\mmorp{\mjrt, \mjrt, \mjrt}$, $\mmorp{\mjrt, \mjrt, \mnrt}$.
Otherwise, $\Gamma$ is intractable.
\end{theorem}

Note that multimorphism $\mmorp{\min, \max}$ corresponds to
submodularity~\cite{Schrijver03:CombOpt}. Constants $\mathcal{C}_D = \{ \rho_0,
\rho_1 \}$ admit multimorphisms $\mmorp{\min, \min}$, $\mmorp{\max, \max}$,
$\mmorp{\min, \max}$, $\mmorp{\mnrt, \mnrt, \mnrt}$, $\mmorp{\mjrt, \mjrt,
\mjrt}$, $\mmorp{\mjrt, \mjrt, \mnrt}$; hence, these six classes of languages
are s-tractable by Lemma~\ref{lmReductionSurjectiveToVCSP}. However, $\rho_0$ does
not admit $\mmorp{c_1}$ and $\rho_1$ does not admit $\mmorp{c_0}$.

\begin{remark}
\label{remGlobalClassificationVCSP}

Although Theorem~\ref{thmClassificationVCSP} is stated only for the weaker notion of
tractability (i.e., for finite languages) in~\cite{cohen06:complexitysoft}, the
proofs there actually establish the same classification for the stronger notion
of global tractability as well.

In particular, all the tractable classes (characterised by the eight
multimorphisms) are globally tractable. Conversely, any globally intractable
language is also intractable: If a language $\Gamma$ does not admit any of the
eight multimorphisms, then there exists a finite subset $\Gamma' \subseteq
\Gamma$ with $|\Gamma'| \leq 8$ that does not admit any of the eight
multimorphisms (since a single weighted relation suffices to violate a
multimorphism).

\end{remark}

We note that Theorem~\ref{thmClassificationVCSP} is a generalisation of
Schaefer's classification of $\{0,\infty\}$-valued constraint
languages~\cite{Schaefer78:complexity} and Creignou's classification of
$\{0,1\}$-valued constraint languages~\cite{Creignou95:jcss}. In particular,
Theorem~\ref{thmClassificationVCSP} implies the following classification of
$\Q$-valued languages.

\begin{theorem}[{\cite[Corollary~7.11]{cohen06:complexitysoft}}]
\label{thmClassificationMinCSP}
Let $\Gamma$ be a Boolean $\Q$-valued language. Then $\Gamma$ is tractable if it
admits any of the following multimorphisms: $\mmorp{c_0}$, $\mmorp{c_1}$,
$\mmorp{\min, \max}$. Otherwise, $\Gamma$ is intractable.
\end{theorem}

Creignou and H{\'{e}}brard~\cite{Creignou97} established a complexity
classification of Boolean $\{0, \infty\}$-valued languages in the surjective
setting.

\begin{theorem}[\cite{Creignou97}]
\label{thmClassificationSurjectiveCSP}

Let $\Gamma$ be a Boolean $\{0, \infty\}$-valued language. Then $\Gamma$ is
s-tractable if it is invariant under any of the following operations: $\min$,
$\max$, $\mnrt$, $\mjrt$. Otherwise, $\Gamma$ is s-intractable.

\end{theorem}

\section{Results}
\label{secSurjectiveResults}

We present our results in three parts: Section~\ref{ssecVCSPsResults} defines the EDS
property and states the main classification theorem,
Section~\ref{ssecFiniteEDSLanguages} focuses on finite EDS languages, and
Section~\ref{secApproximabilityMaxVCSPs} gives a classification in terms of
approximability for the surjective $\MVCSP$.

\subsection{Boolean surjective VCSPs}
\label{ssecVCSPsResults}

We first define the property EDS (which stands for \emph{essentially a downset},
see Definition~\ref{defEssentiallyDownset}) characterising the newly discovered tractable
class of weighted relations.

\begin{definition}
\label{defWRelEDS}

For any $\alpha \geq 1$, an $r$-ary weighted relation $\gamma$ is
\emph{$\alpha$-EDS} if, for every $\tup{x}, \tup{y} \in \Feas(\gamma)$, it holds
$\tup{0}^r \in \Feas(\gamma)$ and
\begin{equation}
\alpha \cdot (\gamma(\tup{x}) + \gamma(\tup{y}) - 2\cdot\gamma(\tup{0}^r)) \geq
\gamma(\sub(\tup{x}, \tup{y})) - \gamma(\tup{0}^r) \,.
\label{eqDefEDSwRel}
\end{equation}
A weighted relation is \emph{EDS} if it is $\alpha$-EDS for some $\alpha \geq
1$. A language is EDS if there exists $\alpha \geq 1$ such that every weighted
relation in the language is $\alpha$-EDS.

\end{definition}

Although this definition does not involve the notion of polymorphisms, it is
stated in a similar vein. Let $h$ be a binary operation defined by $h(x,y) = 0$;
then the requirement ``for every $\tup{x}, \tup{y} \in \Feas(\gamma)$, it holds
$\tup{0}^r \in \Feas(\gamma)$'' translates to ``$\gamma$ is invariant under
$h$'' (or, equivalently, ``$\gamma$ is invariant under $c_0$'').\footnote{In
fact, any EDS weighted relation admits multimorphism $\mmorp{c_0}$ (see
Lemma~\ref{lmEDSAdmitsC0}).}\footnote{Note that the unary empty relation
$\rho_\emptyset$ is vacuously $\alpha$-EDS for all $\alpha \geq 1$, as
$\Feas(\rho_\emptyset) = \emptyset$.} In the case of $\alpha = 1$, inequality
\eqref{eqDefEDSwRel} translates to that of admitting multimorphism
$\mmorp{\sub,h}$. For more intuition behind this notion in the general case, see
the corresponding definition of EDS for set functions (Definition~\ref{defSetFnEDS}) in
Section~\ref{ssecReductionToGMC}. Finite EDS languages admit a simpler equivalent
definition, see Corollary~\ref{corFiniteEDSSub}.

The following classification of $\QInfty$-valued languages is our main result.

\begin{theorem}
\label{thmClassificationSurjectiveVCSP}

Let $\Gamma$ be a Boolean $\QInfty$-valued language. Then $\Gamma$ is globally
s-tractable if it is EDS, or $\neg(\Gamma)$ is EDS, or $\Gamma$ admits any of
the following multimorphisms: $\mmorp{\min, \min}$, $\mmorp{\max, \max}$,
$\mmorp{\min, \max}$, $\mmorp{\mnrt, \mnrt, \mnrt}$, $\mmorp{\mjrt, \mjrt,
\mjrt}$, $\mmorp{\mjrt, \mjrt, \mnrt}$. Otherwise, $\Gamma$ is globally
s-intractable.

\end{theorem}

\begin{proof}

The global s-tractability of languages admitting any of the six multimorphisms
in the statement of the theorem follows from Theorem~\ref{thmClassificationVCSP} (see
Remark~\ref{remGlobalClassificationVCSP}) by Lemma~\ref{lmReductionSurjectiveToVCSP}. The
global s-tractability of EDS languages (whether $\Gamma$ or $\neg(\Gamma)$,
which is symmetric) follows from Theorem~\ref{thmTractabilityEDS}, proved in
Section~\ref{secSurjectiveTractability}. Finally, the global s-intractability of the
remaining languages follows from Theorem~\ref{thmSurjectiveHardness}, proved in
Section~\ref{secSurjectiveHardness}. \qedhere

\end{proof}

\begin{remark}
\label{remLocalClsSrjectiveVCSP}

Theorem~\ref{thmClassificationSurjectiveVCSP} gives us also a classification in terms
of s-tractability. As noted in Section~\ref{ssecWeightedRelations}, any globally
s-tractable language is s-tractable. Consider now a globally s-intractable
language $\Gamma$. It does not admit any of the six multimorphisms, and hence
there exists a finite subset of $\Gamma$ that does not admit them either (see
Remark~\ref{remGlobalClassificationVCSP}). If there exists a finite subset $\Gamma'
\subseteq \Gamma$ such that neither $\Gamma'$ nor $\neg(\Gamma')$ is EDS, then
$\Gamma$ is s-intractable; otherwise $\Gamma$ is s-tractable. Equivalently (by
Corollary~\ref{corFiniteEDSSub}), $\Gamma$ is s-intractable if neither $\Feas(\Gamma)
\cup \Opt(\Gamma)$ nor $\Feas(\neg(\Gamma)) \cup \Opt(\neg(\Gamma))$ admit
polymorphism $\sub$, and it is s-tractable otherwise.

\end{remark}

To see how EDS languages fit into the classification of $\{0, \infty\}$-valued
languages established in Theorem~\ref{thmClassificationSurjectiveCSP}, note the
following. Any $\{0, \infty\}$-valued language of bounded arity is finite. By
Corollary~\ref{corFiniteEDSSub}, any EDS $\{0, \infty\}$-valued language admits
polymorphism $\sub$, and hence also polymorphism $\min$ (by
Lemma~\ref{lmSubImpliesC0Min}).

For $\Q$-valued languages, Theorem~\ref{thmClassificationSurjectiveVCSP} gives a
tighter classification: the only reasons for global s-tractability are EDS and
submodularity.

\begin{theorem}
\label{thmClassificationSurjectiveFiniteValuedVCSP}

Let $\Gamma$ be a Boolean $\Q$-valued language. Then $\Gamma$ is globally
s-tractable if it is EDS, or $\neg(\Gamma)$ is EDS, or $\Gamma$ admits the
$\mmorp{\min,\max}$ multimorphism. Otherwise, $\Gamma$ is globally
s-intractable.

\end{theorem}

\begin{proof}

We need to show that in the case of $\Q$-valued languages, the remaining
globally s-tractable classes from Theorem~\ref{thmClassificationSurjectiveVCSP} (which
are characterised by polymorphisms $\mmorp{\min, \min}$, $\mmorp{\max, \max}$,
$\mmorp{\mnrt, \mnrt, \mnrt}$, $\mmorp{\mjrt, \mjrt, \mjrt}$, and $\mmorp{\mjrt,
\mjrt, \mnrt}$) collapse.

If a $\Q$-valued $r$-ary weighted relation $\gamma$ admits the $\mmorp{\min,
\min}$ multimorphism, then it holds $\gamma(\tup{x}) \geq \gamma(\tup{y})$ for
all $\tup{x} \geq \tup{y}$. This implies that, for all $\tup{x}, \tup{y} \in
\Feas(\gamma)$, it holds $\gamma(\tup{x}) \geq \gamma(\sub(\tup{x}, \tup{y}))$
and $\gamma(\tup{y}) \geq \gamma(\tup{0}^r)$. Hence, $\gamma$ is $1$-EDS. If
$\gamma$ admits the $\mmorp{\max, \max}$ multimorphism, then $\neg(\gamma)$
admits the $\mmorp{\min, \min}$ multimorphism. Therefore, if a $\Q$-valued
language $\Gamma$ admits $\mmorp{\min, \min}$ or $\mmorp{\max, \max}$ as a
multimorphism, then $\Gamma$ or $\neg(\Gamma)$ is EDS.

Multimorphisms $\mmorp{\mnrt, \mnrt, \mnrt}$, $\mmorp{\mjrt, \mjrt, \mjrt}$, and
$\mmorp{\mjrt, \mjrt, \mnrt}$ are covered by the $\mmorp{\min,\max}$
multimorphism: Weighted relations that admit $\mmorp{\mnrt, \mnrt, \mnrt}$ or
$\mmorp{\mjrt, \mjrt, \mjrt}$ as a multimorphism are crisp~\cite[Propositions
6.20 and 6.22]{cohen06:complexitysoft}, and hence, in the $\Q$-valued case, they
are constant functions. $\Q$-valued weighted relations that admit the
$\mmorp{\mjrt, \mjrt, \mnrt}$ multimorphism are modular~\cite[Corollary
6.26]{cohen06:complexitysoft}, and hence they are submodular. \qedhere

\end{proof}

Enumerating all optimal solutions to an instance with polynomial delay is a
fundamental problem~\cite{johnson1988generating,Vazirani1992} studied in the
context of CSP~\cite{Dechter92:aaai,Bulatov12:jcss-enumerating}. An algorithm
outputting a sequence of solutions works with polynomial delay if the time it
takes to output the first solution as well as the time it takes between every
two consecutive solutions is bounded by a polynomial in the input size.

It is known that, for a tractable constraint language $\Gamma$ that includes
constants $\mathcal{C}_D$, one can enumerate all optimal solutions with
polynomial delay~\cite{Cohen04:search}. Our results imply that the newly
discovered globally s-tractable EDS languages enjoy the same property (despite
\emph{not} including constants).

\begin{theorem}
\label{thmEnumerationSurjectiveVCSP}

Let $\Gamma$ be a Boolean $\QInfty$-valued language. If $\Gamma$ is globally
s-tractable then there is a polynomial-delay algorithm that enumerates all
optimal solutions to any instance of $\VCSPs(\Gamma)$.

\end{theorem}

The theorem is proved in Section~\ref{ssecReductionToGMC}.

\subsection{Finite EDS languages}
\label{ssecFiniteEDSLanguages}

The EDS property can be described in a simpler way for languages of finite
size; see the following observation and Corollary~\ref{corFiniteEDSSub}.

\begin{observation}

A language of finite size is EDS if and only if it consists of EDS weighted
relations.

\end{observation}

In the following we prove several useful properties EDS weighted relations.

\begin{lemma}
\label{lmEDSAdmitsC0}

Any EDS weighted relation admits multimorphism $\mmorp{c_0}$.

\end{lemma}

\begin{proof}

Let $\gamma$ be an $r$-ary $\alpha$-EDS weighted relation. For any $\tup{x} \in
\Feas(\gamma)$, it holds $\tup{0}^r \in \Feas(\gamma)$ and
\begin{equation}
\alpha \cdot (2\cdot\gamma(\tup{x}) - 2\cdot\gamma(\tup{0}^r)) \geq
\gamma(\sub(\tup{x}, \tup{x})) - \gamma(\tup{0}^r) = 0
\end{equation}
as $\sub(\tup{x}, \tup{x}) = \tup{0}^r$, and therefore $\gamma(\tup{x}) \geq
\gamma(\tup{0}^r)$. \qedhere

\end{proof}

\begin{lemma}
\label{obsCrispEDS}

A crisp weighted relation is EDS if and only if it admits polymorphism $\sub$.

\end{lemma}

\begin{proof}

Any EDS weighted relation admits polymorphism $\sub$. For the converse
implication, note that any crisp weighted relation that admits polymorphism
$\sub$ (and thus, by Lemma~\ref{lmSubImpliesC0Min}, also polymorphism $c_0$)
satisfies \eqref{eqDefEDSwRel} for any $\alpha \geq 1$. \qedhere

\end{proof}

\begin{lemma}
\label{lmEDSwRelFeasOpt}

A weighted relation $\gamma$ is EDS if and only if both $\Feas(\gamma)$ and
$\Opt(\gamma)$ are EDS.

\end{lemma}

\begin{proof}
Let $\gamma$ be an $r$-ary $\alpha$-EDS weighted relation. For any $\tup{x},
\tup{y} \in \Feas(\gamma)$, it holds $\tup{0}^r \in \Feas(\gamma)$ and
\begin{equation}
\infty >
\alpha \cdot (\gamma(\tup{x}) + \gamma(\tup{y}) - 2\cdot\gamma(\tup{0}^r)) \geq
\gamma(\sub(\tup{x}, \tup{y})) - \gamma(\tup{0}^r) \,,
\end{equation}
and hence $\sub(\tup{x}, \tup{y}) \in \Feas(\gamma)$. By
Lemma~\ref{obsCrispEDS}, $\Feas(\gamma)$ is EDS. Similarly, for any
$\tup{x}, \tup{y} \in \Opt(\gamma)$, it holds $\tup{0}^r \in \Opt(\gamma)$ (by
Lemma~\ref{lmEDSAdmitsC0}) and
\begin{equation}
0 =
\alpha \cdot (\gamma(\tup{x}) + \gamma(\tup{y}) - 2\cdot\gamma(\tup{0}^r)) \geq
\gamma(\sub(\tup{x}, \tup{y})) - \gamma(\tup{0}^r) \,;
\end{equation}
therefore $\sub(\tup{x}, \tup{y}) \in \Opt(\gamma)$ and $\Opt(\gamma)$ is EDS.

To prove the converse implication, let us assume that $\Feas(\gamma)$,
$\Opt(\gamma)$ are EDS and consider any $\tup{x}, \tup{y} \in \Feas(\gamma)$. As
$\Opt(\gamma)$ admits polymorphism $c_0$, it holds $\tup{0}^r \in \Opt(\gamma)
\subseteq \Feas(\gamma)$. Therefore, the left-hand side of~\eqref{eqDefEDSwRel}
is non-negative. Moreover, if it equals $0$, then $\tup{x}, \tup{y} \in
\Opt(\gamma)$, and hence $\sub(\tup{x}, \tup{y}) \in \Opt(\gamma)$ and the
right-hand side equals $0$ as well. Therefore, \eqref{eqDefEDSwRel} holds for
large enough $\alpha$, as there are only finitely many choices of $\tup{x},
\tup{y} \in \Feas(\gamma)$.
\end{proof}

We show that relations invariant under $\sub$ have a simple structure.

\begin{definition}
\label{defEssentiallyDownset}
An $r$-ary relation $\rho$ is a \emph{downset} if, for any $r$-tuples $\tup{x},
\tup{y}$ such that $\tup{x} \geq \tup{y}$  and $\tup{x} \in \rho$, it holds
$\tup{y} \in \rho$.

An $r$-ary relation $\rho$ is \emph{essentially a downset} if it can be written
as a conjunction of a downset and binary equality relations. Formally, there
exists a downset $\rho'$ with $\ar(\rho') = r' \leq r$, a permutation $\pi$ of
$[r]$, and indices $a_{r'+1}, \dots, a_r \in \{\pi(1), \dots, \pi(r')\}$ such
that
\begin{equation}
\rho(x_1, \dots, x_r) =
\rho'\left(x_{\pi(1)}, \dots, x_{\pi(r')}\right) +
\sum_{i=r'+1}^r \rho_=\left(x_{\pi(i)},x_{a_{i}}\right) \,.
\end{equation}
(Note that addition of crisp weighted relations corresponds to conjunction.) In
other words, removing duplicate coordinates\footnote{A coordinate $i$ is a
duplicate of a coordinate $j$ if, for every $(x_1, \dots, x_r) \in \rho$, it
holds $x_i = x_j$.} of $\rho$ results in a downset.

\end{definition}

\begin{example}

Relation $\rho' = \{(0,0), (0,1), (1,0)\}$ is a downset, while $\rho =
\{(0,0,0), (0,1,1), (1,0,0)\}$ is only essentially a downset (as $\rho(x,y,z) =
\rho'(x,y) + \rho_=(y,z)$).

\end{example}

\begin{lemma}
\label{lmAlgebraicEDS}

A relation is essentially a downset if and only if it admits polymorphism
$\sub$.

\end{lemma}

\begin{proof}

For any $r$-ary relation $\rho$ that is essentially a downset and $\tup{x},
\tup{y} \in \rho$, we prove that $\tup{z} = \sub(\tup{x}, \tup{y}) \in \rho$.
Let $\tup{x} = (x_1, \dots, x_r)$, $\tup{y} = (y_1, \dots, y_r)$, $\tup{z} =
(z_1, \dots, z_r)$. It holds $\tup{x} \geq \tup{z}$. Moreover, for any
coordinates $i, j$ such that $x_i = x_j$ and $y_i = y_j$, it holds $z_i = z_j$.
Since $\rho$ can be written as a sum of a downset and equality relations, we
have $\tup{z} \in \rho$.

We prove the converse implication by contradiction. Suppose that $\rho$ is a
smallest-arity relation that admits polymorphism $\sub$ but is not essentially a
downset; let us denote its arity by $r$. If there are distinct coordinates $i,
j$ such that $z_i = z_j$ for all $\tup{z} = (z_1, \dots, z_r) \in \rho$,
identifying these coordinates yields an $(r-1)$-ary relation $\rho'$ such that
$\rho$ can be written as the sum of $\rho'$ and a binary equality relation.
However, $\rho'$ also admits $\sub$, and hence is essentially a downset by the
choice of $\rho$, which implies that $\rho$ is essentially a downset as well.
Therefore, for any distinct coordinates $i, j$, there exists $\tup{z}^{(i,j)}
\in \rho$ with $z^{(i,j)}_i \neq z^{(i,j)}_j$.

As $\rho$ is not a downset, for some $r$-tuples $\tup{x}, \tup{y}$ it holds
$\tup{x} \geq \tup{y}$, $\tup{x} \in \rho$, $\tup{y} \not\in \rho$. We may
assume without loss of generality that, for some $n \in [r]$, the set of
coordinates with label $1$ equals $[n]$ for $\tup{x}$ and $[n-1]$ for $\tup{y}$.
Let $\tup{e} = (e_1, \dots, e_r) \in \rho$ be a tuple with the smallest number
of coordinates labelled $1$ such that $e_n = 1$. We claim that $e_i = 0$ for all
$i \neq n$: Otherwise, either $\sub\left(\tup{e}, \tup{z}^{(i,n)}\right) =
\min\left(\tup{e}, \neg(\tup{z}^{(i,n)})\right)$ or $\sub\left(\tup{e},
\sub\left(\tup{e}, \tup{z}^{(i,n)}\right)\right) = \min\left(\tup{e},
\tup{z}^{(i,n)}\right)$ contradicts the minimality of $\tup{e}$. But then
$\sub(\tup{x}, \tup{e}) = \tup{y} \in \rho$, which is a contradiction. \qedhere

\end{proof}

\begin{corollary}
\label{corFiniteEDSSub}

Let $\Gamma$ be a finite language. The following conditions are equivalent.

\begin{enumerate}

\item Language $\Gamma$ is EDS.

\item For every $\gamma \in \Gamma$, weighted relation $\gamma$ is EDS.

\item For every $\gamma \in \Gamma$, both $\Feas(\gamma)$ and $\Opt(\gamma)$
admit polymorphism $\sub$.

\item For every $\gamma \in \Gamma$, both $\Feas(\gamma)$ and $\Opt(\gamma)$
are essentially downsets.

\end{enumerate}

\end{corollary}

\begin{remark}
\label{remAlmostMinMinEDSEquivalence}

In~\cite{fz17:mfcs}, a weighted relation $\gamma$ is called PDS if both
$\Feas(\gamma)$ and $\Opt(\gamma)$ are essentially downsets. For a
$\{0,1\}$-valued weighted relation, this condition is equivalent to that of
being almost-min-min~\cite{Uppman12:cp}. By Corollary~\ref{corFiniteEDSSub}, PDS and EDS
are equivalent concepts for languages of finite size.

\end{remark}

As we show in the following example, there exists an infinite non-EDS language
$\Gamma$ such that every finite subset $\Gamma' \subseteq \Gamma$ is EDS. Hence,
$\Gamma$ is s-tractable, although it is globally s-intractable ($\VCSPs(\Gamma)$
is NP-hard).

\begin{example}
\label{exGloballySIntractable}

For any $w\in\mathbb{Z}_{\geq 1}$, we define a ternary weighted relation
$\mu_w$ on $D = \{0, 1\}$ by
\begin{equation}
\mu_w(x,y,z) =
\begin{cases}
2 & \text{if $z=1$ and $x=y$,} \\
1 & \text{if $z=1$ and $x\neq y$,} \\
0 & \text{if $z=0$ and $x=y=0$,} \\
w & \text{otherwise.} \\
\end{cases}
\end{equation}
Note that $\Feas(\mu_w) = D^3$ and $\Opt(\mu_w) = \{ (0,0,0) \}$ are
downsets, and hence $\mu_w$ is EDS. However, it is not $\alpha$-EDS for any
$\alpha < w/2$: For $\tup{x} = (0,1,1)$, $\tup{y} = (1,0,1)$, we have
$\mu_w(\tup{x}) + \mu_w(\tup{y}) = 2$ but $\mu_w(\sub(\tup{x},
\tup{y})) = \mu_w(0,1,0) = w$. Language $\Gamma = \left\{\mu_w ~|~
w\in\mathbb{Z}_{\geq 1}\right\}$ is therefore not EDS.

By our classification (Theorem~\ref{thmClassificationSurjectiveVCSP}), language
$\Gamma$ is globally s-intractable; here we show it directly by a reduction from
the NP-hard Max-Cut problem. Given an undirected graph $G = (V,E)$ with no
isolated vertices, we construct a $\VCSPs(\Gamma)$ instance $I$ as follows. Let
$w = 2|E| + 1$. We introduce a corresponding variable for every vertex in $V$,
and add a special variable $z$. For every edge $\{x, y\} \in E$, we impose a
constraint $\mu_w(x,y,z)$.

Cuts in $G$ are in one-to-one correspondence with assignments to $I$ satisfying
$z=1$. In particular, a cut of size $k$ corresponds to an assignment to $I$ with
value $k+2(|E|-k)=2|E|-k$. Any surjective assignment with $z = 0$ is of value at
least $w > 2|E|-k$. Thus, solving $I$ amounts to solving Max-Cut in $G$.

\end{example}

\subsection{Approximability of maximising surjective VCSP}
\label{secApproximabilityMaxVCSPs}

Although the VCSP is commonly defined with a minimisation objective, it is easy
to see that, for exact solvability, its maximisation variant is essentially an
identical problem: Minimising a $\Q$-valued function $\phi_I$ corresponds to
maximising $-\phi_I$. When studying approximability, however, the two variants
vastly differ (see~\cite{makarychev2017approximation} for a survey).

We focus on maximisation of the $\Qn$-valued VCSP. This problem generalises the
Max-CSP, in which the objective is to maximise the number of satisfied
constraints; in particular, the Max-CSP corresponds to maximisation of the
$\{0,1\}$-valued VCSP. The complexity of exactly maximising the $\Qn$-valued
VCSP was established by Thapper and \v{Z}ivn\'y~\cite{tz16:jacm}.
Raghavendra~\cite{raghavendra08:stoc} showed that, assuming the unique games
conjecture, the basic semidefinite programming relaxation achieves the optimal
approximation ratio for the problem. In this section, we consider approximate
maximisation of the \emph{surjective} $\Qn$-valued VCSP.

\begin{definition}

An \emph{instance} $I=(V,D,\phi_I)$ of the Max-VCSP on domain $D$ is given by a
finite set of variables $V=\{x_1,\ldots,x_n\}$ and an objective function
$\phi_I:D^n\to\Qn$ expressed as a weighted sum of \emph{constraints} over
$V$, i.e.,
\begin{equation}
\phi_I(x_1,\ldots,x_n)=\sum_{i=1}^q w_i\cdot\gamma_i(\tup{x}_i) \,,
\end{equation}
where $\gamma_i$ is a $\Qn$-valued weighted relation, $w_i \in \Qn$ is the
\emph{weight} and $\tup{x}_i\in V^{\ar(\gamma_i)}$ the \emph{scope} of the $i$th
constraint.

Given an instance $I$, the goal is to find an assignment $s:V\to D$ of domain
labels to the variables that \emph{maximises} $\phi_I$. We denote the maximum
value of the objective function by $\opt_I$. For any $r \in [0, 1]$, an
assignment $s$ is an \emph{$r$-approximate solution} to $I$ if $\phi_I(s) \geq
r\cdot\opt_I$.

An assignment $s$ is \emph{surjective} if its image equals $D$. We denote the
maximum objective value of surjective assignments by $\sopt_I$. For any $r \in
(0, 1]$, a surjective assignment $s$ is an \emph{$r$-approximate surjective
solution} to $I$ if $\phi_I(s) \geq r\cdot\sopt_I$.

We denote by $\MVCSPs(\Gamma)$ the surjective $\MVCSP$ problem on instances over
a language $\Gamma$.

\end{definition}

Following the standard definitions, we say that $\MVCSPs(\Gamma)$ belongs to APX
if, for some $r \in (0, 1]$, there exists a polynomial-time algorithm that finds
an $r$-approximate surjective solution to every $\MVCSPs(\Gamma)$ instance. If
such an algorithm exists for every $r < 1$, we say that the problem admits a
polynomial-time approximation scheme (PTAS). $\MVCSPs(\Gamma)$ is APX-hard if
there exists a PTAS reduction (an approximation-preserving reduction,
see~\cite{crescenzi1997guide}) from every problem in APX to $\MVCSPs(\Gamma)$.

First, we prove that a polynomial-time algorithm for exactly maximising the
$\Qn$-valued VCSP over a language $\Gamma$ implies a PTAS for $\MVCSPs(\Gamma)$.
Second, we establish a complexity classification of Boolean languages in
Theorem~\ref{thmApproximabilitySurjectiveFiniteValuedVCSP}.

\begin{lemma}
\label{lmSurjectiveApproximation}

Let $\Gamma$ be a $\Qn$-valued language and $r, \epsilon \in \R$ such that $0 <
\epsilon \leq r \leq 1$. There is a polynomial-time algorithm that, given a
Max-VCSP instance $I=(V,D,\phi_I)$ over $\Gamma$ and an $r$-approximate solution
$s$ to $I$, outputs an $(r-\epsilon)$-approximate \emph{surjective} solution
$s'$ to $I$.

\end{lemma}

\begin{proof}

Let $a_\text{max}$ denote the maximum arity of weighted relations in $\Gamma$,
and $n$ the number of variables of $I$. If $n < \frac{r \cdot |D| \cdot
a_\text{max}}{\epsilon}$, we find an optimal surjective assignment to $I$ by
trying all $O(|D|^n)$ assignments.

Otherwise, we modify the given assignment $s$ in order to obtain a surjective
assignment $s'$. For any variable $x \in V$, let $B_x \subseteq [q]$ be the set
of indices of constraints in whose scopes $x$ appears. We define the
contribution of $x$ by
\begin{equation}
c(x) = \sum_{i \in B_x} w_i\cdot\gamma_i(s(\tup{x}_i)) \,.
\end{equation}
It follows that the total contribution of all variables is at most $a_\text{max}
\cdot \phi_I(s)$.

Let $U$ be a set of $|D|$ variables with the smallest contribution. We
assign to them labels $D$ bijectively. The resulting assignment $s'$ is
surjective, and it holds
\begin{align}
\phi_I(s')
&\geq \phi_I(s) - \sum_{x\in U} c(x) \\
&\geq \phi_I(s) - \frac{|D|}{n} \cdot a_\text{max} \cdot \phi_I(s) \\
&\geq \left(1- \frac{|D|}{n}\cdot a_\text{max} \right) \cdot r \cdot \opt_I \\
&\geq (r-\epsilon) \cdot \sopt_I \,.
\end{align}
\qedhere

\end{proof}

Applying this lemma to an \emph{optimal} solution to an $\MVCSP$ instance (i.e.,
$r = 1$) gives us the following corollary.

\begin{corollary}
\label{corPTASforSurjectiveMVCSP}

If the $\MVCSP$ over a $\Qn$-valued language $\Gamma$ is solvable in polynomial
time, then there is a PTAS for $\MVCSPs(\Gamma)$.

\end{corollary}

Finally, we classify Boolean $\Qn$-valued languages by the complexity of the
corresponding $\MVCSPs$. Since multimorphisms and the EDS property are defined
in the context of minimisation, the following theorem applies them to language
$-\Gamma$ instead of $\Gamma$ (where $-\Gamma = \{ -\gamma ~|~ \gamma \in
\Gamma\}$ and $(-\gamma)(\tup{x}) = - \gamma(\tup{x})$).

\begin{theorem}
\label{thmApproximabilitySurjectiveFiniteValuedVCSP}

Let $\Gamma$ be a Boolean $\Qn$-valued language. Then

\begin{enumerate}

\item \label{it1} $\MVCSPs(\Gamma)$ is solvable exactly in polynomial time if
$-\Gamma$ is EDS, or $-(\neg(\Gamma))$ is EDS, or $-\Gamma$ admits the $\mmorp{\min,
\max}$ multimorphism;

\item \label{it2}
otherwise it is NP-hard to solve exactly, but

\begin{enumerate}

\item \label{it2a}
it is in PTAS if $-\Gamma$ admits $\mmorp{c_0}$ or $\mmorp{c_1}$,

\item \label{it2b}
and is APX-hard otherwise.

\end{enumerate}

\end{enumerate}

\end{theorem}

\begin{proof}

Theorem~\ref{thmClassificationSurjectiveFiniteValuedVCSP} implies the case~(\ref{it1})
and NP-hardness in the case~(\ref{it2}). Case~(\ref{it2a}) follows from
Theorem~\ref{thmClassificationMinCSP} and Corollary~\ref{corPTASforSurjectiveMVCSP}. By
Theorem~\ref{thmClassificationMinCSP}, if $-\Gamma$ does not admit either of
$\mmorp{c_0}$, $\mmorp{c_1}$ and $\mmorp{\min, \max}$, then $\MVCSP(\Gamma)$ is
NP-hard. The proof of Theorem~\ref{thmClassificationMinCSP}
in~\cite{cohen06:complexitysoft} actually establishes that $\MVCSP(\Gamma)$ is
APX-hard. By the approximation-preserving reduction in the proof of
Lemma~\ref{lmReductionVCSPtoSurjective}, this implies that $\MVCSPs(\Gamma)$ is
APX-hard as well. \qedhere

\end{proof}

Theorem~\ref{thmApproximabilitySurjectiveFiniteValuedVCSP} generalises the result of
Bach and Zhou~\cite[Theorem~16]{Bach11:surj} in two respects. Firstly, we
classify all $\Qn$-valued languages as opposed to $\{0,1\}$-valued languages.
Secondly, we classify constraint languages as being in P, in PTAS, or being
APX-hard; \cite{Bach11:surj} only distinguishes admitting a PTAS versus being
APX-hard. Finally, the main technical component of
Theorem~\ref{thmApproximabilitySurjectiveFiniteValuedVCSP},
Lemma~\ref{lmSurjectiveApproximation}, has a slightly simpler proof compared
to~\cite{Bach11:surj}.

\section{Hardness proofs}
\label{secSurjectiveHardness}

Consider a Boolean language $\Gamma$ over $D = \{ 0, 1 \}$ that admits
multimorphism $\mmorp{c_0}$ (the case of multimorphism $\mmorp{c_1}$ is
symmetric), but does not admit any of the following multimorphisms:
$\mmorp{\min, \min}$, $\mmorp{\max, \max}$, $\mmorp{\min, \max}$, $\mmorp{\mnrt,
\mnrt, \mnrt}$, $\mmorp{\mjrt, \mjrt, \mjrt}$, $\mmorp{\mjrt, \mjrt, \mnrt}$.
Suppose that $\Gamma$ is not EDS. We prove that $\VCSPs(\Gamma)$ is NP-hard,
i.e., $\Gamma$ is globally s-intractable.

We start by showing that there exists a relation such that it is not invariant
under $\sub$ and it can be added to $\Gamma$ without changing the complexity of
$\VCSPs(\Gamma)$ (see Corollary~\ref{corNotEDSLanguageReduction}). For finite $\Gamma$,
this follows simply from Corollary~\ref{corFiniteEDSSub} and
Lemma~\ref{lmSurjectiveClosureReduction}, as there exists $\gamma \in \Gamma$ such
that $\Feas(\gamma)$ or $\Opt(\gamma)$ is not invariant under $\sub$. In
general, however, a different argument is necessary. We prove it by showing that
$\Gamma$ contains weighted relations arbitrarily ``similar'' to a relation which
is not invariant under $\sub$, and that this relation may be thus added to
$\Gamma$.

\begin{definition}

For any $\alpha \geq 1$, an $r$-ary weighted relation $\gamma$ is
\emph{$\alpha$-crisp} if its image $\gamma(D^r)$ lies in ${[0,1]} \cup {(\alpha,
\infty]}$. We will denote by $\Round{\alpha}{\gamma}$ the $r$-ary relation
defined as
\begin{equation}
\Round{\alpha}{\gamma}(\tup{x}) =
\begin{cases}
0      & \text{if $\gamma(\tup{x}) \in {[0,1]}$,} \\
\infty & \text{if $\gamma(\tup{x}) \in {(\alpha, \infty]}$.}
\end{cases}
\end{equation}

\end{definition}

Note that an $\alpha$-crisp weighted relation is $\alpha'$-crisp for any
$\alpha' \leq \alpha$. Moreover, a crisp weighted relation $\rho$ is $\alpha$-crisp for
any $\alpha \geq 1$, and $\Round{\alpha}{\rho} = \rho$.

\begin{lemma}
\label{lmAlphaCrispReduction}

Let $\Gamma$ be a language and $\rho$ a relation such that, for any $\alpha \geq
1$, there exists an $\alpha$-crisp weighted relation $\gamma \in \Gamma$ with
$\Round{\alpha}{\gamma} = \rho$. Then $\VCSPs(\Gamma \cup \{\rho\}) \reduces
\VCSPs(\Gamma)$.

\end{lemma}

\begin{proof}

Let $I$ be an instance of $\VCSPs(\Gamma \cup \{\rho\})$ with $k$ constraints
that apply relation $\rho$. By scaling and adding rational constants to weighted
relations in $I$, we ensure that all the assigned values are non-negative
integers. Let $M$ be an upper bound on the maximum value of a feasible solution
to $I$ (e.g., the weighted sum of the maximum finite values assigned by the
constraints of $I$). Let $\gamma \in \Gamma$ be a $M\cdot(k+1)$-crisp weighted
relation such that $\Round{M\cdot(k+1)}{\gamma} = \rho$. In each constraint
applying relation $\rho$, we replace it by $\gamma$ with weight $1/(k+1)$, and
thus obtain an instance of $\VCSPs(\Gamma)$. Since $\gamma$ is
$M\cdot(k+1)$-crisp, the value of any feasible assignment increases by at most
$k/(k+1) < 1$, and the value of any infeasible assignment becomes larger than
$M$. \qedhere

\end{proof}

\begin{lemma}
\label{lmNotEDSLanguageToRelation}

Let $\Gamma$ be a language such that it admits multimorphism $\mmorp{c_0}$ but
is not EDS. Then there exists a relation $\rho$ that is invariant under $c_0$
but not under $\sub$ and, for any $\alpha \geq 1$, there exists an
$\alpha$-crisp weighted relation $\gamma \in \clGamma$ with
$\Round{\alpha}{\gamma} = \rho$.

\end{lemma}

\begin{proof}

We will show that for any $\alpha \geq 1$, there exists an $\alpha$-crisp
weighted relation $\gamma \in \clGamma$ such that $\Round{\alpha}{\gamma}$ is a
relation of arity at most $4$ that is invariant under $c_0$ but not under
$\sub$. As there are only finitely many such relations, the claim of the lemma
will follow.

Language $\clGamma$ admits multimorphism $\mmorp{c_0}$ as well but is not EDS;
in particular, it is not $\alpha^{17}$-EDS. Therefore, there exists an $r$-ary
weighted relation $\gamma \in \clGamma$ and $\tup{u}, \tup{v} \in \Feas(\gamma)$
such that $\gamma(\tup{0}^r) = 0$ (as $\clGamma$ is closed under adding rational
constants) and
\begin{equation}
0 \leq
\alpha^{17} \cdot (\gamma(\tup{u}) + \gamma(\tup{v})) <
\gamma(\sub(\tup{u}, \tup{v})) \,.
\end{equation}
We may assume that there are no distinct coordinates $i,j$ where $u_i=u_j$ and
$v_i=v_j$ (otherwise we identify them), and hence $r \leq 4$. As $\clGamma$ is
closed under scaling, we may also assume that $\gamma(\tup{u}),\gamma(\tup{v})
\leq 1$ and $\gamma(\sub(\tup{u}, \tup{v})) > \alpha^{17}$.

Let us consider, for any $0 \leq i \leq 16$, the intersection of the image
$\gamma(D^r)$ with the interval ${\left(\alpha^i, \alpha^{i+1}\right]}$. Since
$|D^r| \leq 2^4 = 16$, the intersection is empty for some $i$. Scaling $\gamma$
by $1 / \alpha^i$ then yields an $\alpha$-crisp weighted relation $\gamma' \in
\clGamma$ such that $\Round{\alpha}{\gamma'}$ is invariant under $c_0$ but not
under $\sub$, as $\gamma'(\tup{0}^r), \gamma'(\tup{u}), \gamma'(\tup{v}) \leq 1$
and $\gamma'(\sub(\tup{u}, \tup{v})) > \alpha$. \qedhere

\end{proof}

\begin{corollary}
\label{corNotEDSLanguageReduction}

Let $\Gamma$ be a language such that it admits multimorphism $\mmorp{c_0}$ but
is not EDS. Then $\VCSPs(\Gamma \cup \{\rho\}) \reduces \VCSPs(\Gamma)$ for some
relation $\rho$ that is invariant under $c_0$ but not under $\sub$.

\end{corollary}

\begin{proof}

By Lemmas~\ref{lmNotEDSLanguageToRelation}~and~\ref{lmAlphaCrispReduction}, we have that
$\VCSPs(\clGamma \cup \{\rho\}) \reduces \VCSPs(\clGamma)$ for some relation
$\rho$ that is invariant under $c_0$ but not under $\sub$. By
Lemma~\ref{lmSurjectiveClosureReduction}, it holds $\VCSPs(\clGamma) \reduces
\VCSPs(\Gamma)$. \qedhere

\end{proof}

We define weighted relations $\gamma_0 = \Soft(\rho_0)$, $\gamma_1 =
\Soft(\rho_1)$, and $\gamma_= = \Soft(\rho_=)$; a binary relation
$\rho_\leq=\{(0,0),(0,1),(1,1)\}$, and, for $r\in\{3,4\}$, an $r$-ary relation
\begin{equation}
A_r
=
\left\{ (x_1, \dots, x_r) \in \{0,1\}^r
~\middle|~
\sum_{i=1}^r x_i \equiv 0 \, (\mbox{mod } 2) \right\} \,.
\end{equation}

Assuming that $\Gamma$ does not admit polymorphism $\sub$, we prove that
$\VCSPs(\Gamma)$ is NP-hard (see Lemma~\ref{lmNotEDSsIntractable}). The proof makes
use of several sources of hardness. More specifically, we show that at least one
of the following cases applies:

\begin{itemize}

\item $\VCSPs(\Feas(\Gamma) \cup \Opt(\Gamma))$ is NP-hard (by the
classification of $\{ 0, \infty \}$-valued languages, see
Theorem~\ref{thmClassificationSurjectiveCSP}).

\item $\VCSP(\Gamma \cup \mathcal{C}_D)$ reduces to $\VCSPs(\Gamma)$. In
particular, it holds $\rho_\leq \in \clGamma$, which can be used to simulate
constants (see Lemma~\ref{lmReductionLeq}). The intractability of $\VCSP(\Gamma \cup
\mathcal{C}_D)$ follows from Theorem~\ref{thmClassificationVCSP}.

\item The NP-hard Minimum Distance problem~\cite{Vardy1997} reduces to
$\VCSPs(\Gamma)$. In particular, it holds $\{ A_3, \gamma_0 \} \subseteq
\clGamma$ or $\{ A_4, \gamma_= \} \subseteq \clGamma$; the reduction from the
Minimum Distance problem to these languages is given in
Lemma~\ref{lmA34sIntractable}.

\end{itemize}

Before proving Lemma~\ref{lmNotEDSsIntractable}, we need a few auxiliary lemmas to
establish the existence of certain weighted relations in $\clGamma$.

\begin{lemma}
\label{lmNotNeg}
Let $\rho$ be a relation invariant under $c_0$ but not under $\neg$.
Then $\rho_0 \in \clSingle{\rho}$ or $\rho_\leq \in \clSingle{\rho}$.
\end{lemma}

\begin{proof}
Let $r$ denote the arity of $\rho$. There exists an $r$-tuple $\tup{u} \in \rho$
such that $\neg(\tup{u}) \not\in \rho$. If $\tup{1}^r \not\in \rho$, we obtain
$\rho_0$ by identifying all coordinates of $\rho$. Otherwise, we obtain
$\rho_\leq$ by identifying all coordinates where $u_i = 0$ and identifying all
coordinates where $u_i = 1$.
\end{proof}

\begin{lemma}
\label{lmNonCrisp}
Let $\gamma$ be a non-crisp weighted relation such that it admits multimorphism
$\mmorp{c_0}$. Then $\gamma_0 \in \clSingle{\gamma, \rho_0}$. If in addition $\Feas(\gamma)$
and $\Opt(\gamma)$ are invariant under $\neg$, then $\gamma_= \in
\clSingle{\gamma}$.
\end{lemma}

\begin{proof}
Let $r$ denote the arity of $\gamma$. There exists an $r$-tuple $\tup{u}$ such
that $\gamma(\tup{0}^r) < \gamma(\tup{u}) < \infty$. By $0$-pinning at all
coordinates where $u_i = 0$ and identifying all coordinates where $u_i = 1$, we
obtain a unary weighted relation $\gamma' \in \clSingle{\gamma, \rho_0}$ such
that $\gamma'(0) < \gamma'(1) < \infty$. From it, we can obtain $\gamma_0$ by
adding a rational constant and scaling, as $\gamma_0 = \frac{\gamma' -
\gamma'(0)}{\gamma'(1) - \gamma'(0)}$.

If $\Feas(\gamma)$ and $\Opt(\gamma)$ are invariant under $\neg$, it holds
$\gamma(\tup{1}^r) = \gamma(\tup{0}^r)$ and $\gamma(\tup{0}^r) <
\gamma(\neg(\tup{u})) < \infty$. By identifying all coordinates where $u_i = 0$
and identifying all coordinates where $u_i = 1$, we obtain a binary weighted
relation $\gamma' \in \clSingle{\gamma}$. Consider $\gamma'' \in
\clSingle{\gamma}$ defined as $\gamma''(x, y) = \gamma'(x, y) + \gamma'(y, x)$.
It holds $\gamma''(0,0) = \gamma''(1,1) < \gamma''(0,1) = \gamma''(1,0) <
\infty$. From it, we can obtain $\gamma_=$ by adding a rational constant and
scaling.
\end{proof}

\begin{lemma}
\label{lmNegMnrtNotEDS}
Let $\rho$ be a relation invariant under $c_0$, $\neg$, and $\mnrt$, but not
under $\sub$. Then $A_4 \in \clSingle{\rho}$.
\end{lemma}

\begin{proof}
Let $\rho'$ be a smallest-arity relation in $\clSingle{\rho}$ that is not
invariant under $\sub$, and denote its arity by $r$. As $\tup{0}^r \in \rho'$
and $\mnrt(\tup{x}, \tup{y}, \tup{0}^r) = \tup{x} \oplus \tup{y}$, relation
$\rho'$ is closed under the $\oplus$ operation. Let $\tup{u}, \tup{v} \in \rho'$
be $r$-tuples such that $\sub(\tup{u}, \tup{v}) \not\in \rho'$. There are no
distinct coordinates $i,j$ where $u_i=u_j$ and $v_i=v_j$, otherwise we could
identify them to obtain an $(r-1)$-ary relation not invariant under $\sub$. For
any $a, b \in \{0, 1\}$, there is a coordinate $i$ where $u_i = a$ and $v_i =
b$, otherwise $\sub(\tup{u}, \tup{v})$ would be equal to $\neg(\tup{v})$,
$\tup{u} \oplus \tup{v}$, $\tup{0}^r$, or $\tup{u}$ respectively, which would
imply $\sub(\tup{u}, \tup{v}) \in \rho'$. Therefore, $r = 4$, and we may assume
without loss of generality that $\tup{u} = (0,0,1,1)$, $\tup{v} = (0,1,0,1)$. As
\begin{align*}
\sub(\tup{u}, \tup{v})
&= (0,0,1,0) \\
&= (0,0,0,1) \oplus \tup{u} \\
&= (0,1,0,0) \oplus (\tup{u} \oplus \tup{v}) \\
&= (1,0,0,0) \oplus \neg(\tup{v}) \,,
\end{align*}
it holds $(0,0,0,1), (0,0,1,0), (0,1,0,0), (1,0,0,0) \not\in \rho'$.
Since $\rho'$ is closed under $\neg$, we have $\rho' = A_4$.
\end{proof}

\begin{lemma}
\label{lmNotNegEDS}
Let $\rho$ be a relation invariant under $c_0$ but not under $\sub$.
If $\rho$ is invariant under $\mnrt$, then $A_3 \in \clSingle{\rho, \rho_0}$. If
$\rho$ is invariant under $\min$ or $\max$, then $\rho_\leq \in \clSingle{\rho,
\rho_0}$.
\end{lemma}

\begin{proof}
Let $\rho'$ be a smallest-arity relation in $\clSingle{\rho, \rho_0}$ that is
not invariant under $\sub$, and denote its arity by $r$. Let $\tup{u}, \tup{v}
\in \rho'$ be $r$-tuples such that $\sub(\tup{u}, \tup{v}) \not\in \rho'$. There
are no distinct coordinates $i,j$ where $u_i=u_j$ and $v_i=v_j$, otherwise we
could identify them to obtain an $(r-1)$-ary relation not invariant under
$\sub$. For any $b \in \{0, 1\}$, there is a coordinate $i$ where $u_i = 1$ and
$v_i = b$, otherwise $\sub(\tup{u}, \tup{v})$ would be equal to $\tup{0}^r$ or
$\tup{u}$ respectively, which would imply $\sub(\tup{u}, \tup{v}) \in \rho'$.
However, there is no coordinate $i$ where $u_i=v_i=0$, otherwise we could obtain
an $(r-1)$-ary relation not invariant under $\sub$ by $0$-pinning $\rho'$ at
coordinate $i$. Therefore, $r = 2$ or $r = 3$. If $r = 2$, we have $\rho_\leq
\in \clSingle{\rho, \rho_0}$, and $\rho$ is not invariant under $\mnrt$ (as
neither is $\rho_\leq$).

If $r = 3$, we may assume without loss of generality that $\tup{u} = (0,1,1)$
and $\tup{v} = (1,0,1)$. Relation $\rho$ is not invariant under $\min$,
otherwise it would hold $\min(\tup{u}, \tup{v}) = (0,0,1) \in \rho'$ and we
could obtain a binary relation not invariant under $\sub$ by $0$-pinning $\rho'$
at the first coordinate. Similarly, relation $\rho$ is not invariant under
$\max$, otherwise it would hold $\max(\tup{u}, \tup{v}) = (1,1,1) \in \rho'$ and
we could obtain a binary relation not invariant under $\sub$ by identifying the
first and third coordinate. Finally, assume that relation $\rho$ is invariant
under $\mnrt$. Then $\rho'$ is also closed under the $\oplus$ operation, as
$\mnrt(\tup{x}, \tup{y}, \tup{0}^r) = \tup{x} \oplus \tup{y}$, and we have
$\tup{u} \oplus \tup{v} = (1,1,0) \in \rho'$. Since $\sub(\tup{u}, \tup{v}) =
(0,1,0) = (0,0,1) \oplus \tup{u} = (1,1,1) \oplus \tup{v} = (1,0,0) \oplus
(\tup{u} \oplus \tup{v})$, it holds $(0,0,1), (1,1,1), (1,0,0) \not\in \rho'$,
and therefore $\rho' = A_3$.
\end{proof}

\begin{lemma}
\label{lmReductionLeq}

If $\rho_\leq \in \Gamma$, then $\VCSP(\Gamma \cup \mathcal{C}_D) \reduces
\VCSPs(\Gamma)$.

\end{lemma}

\begin{proof}

For a given instance of $\VCSP(\Gamma \cup \{\rho_0, \rho_1\})$ with variables
$V$, we construct an instance of $\VCSPs(\Gamma)$ as follows: We introduce new
variables $y_0$, $y_1$ and impose constraints $\rho_\leq(y_0, x)$, $\rho_\leq(x,
y_1)$ for all $x \in V$ to ensure that $y_0 = 0$, $y_1 = 1$ in any feasible
surjective assignment. Then we replace each constraint of the form $\rho_0(x)$
with $\rho_\leq(x, y_0)$ and each constraint of the form $\rho_1(x)$ with
$\rho_\leq(y_1, x)$. \qedhere

\end{proof}

\begin{lemma}
\label{lmA34sIntractable}

Languages $\{A_3, \gamma_0\}$ and $\{A_4, \gamma_=\}$ are both s-intractable.

\end{lemma}

\begin{proof}

First we show a reduction from the optimisation variant of the Minimum Distance
problem, which is NP-hard \cite{Vardy1997}, to $\VCSPs(\{A_3, \gamma_0\})$. A
problem instance is given as an $m \times n$ matrix $H$ over the field $D = \{0,
1\}$, and the objective is to find a non-zero vector $\tup{x} = (x_1, \dots,
x_n) \in D^n$ satisfying $H \cdot \tup{x} = \tup{0}^m$ with the minimum weight
(i.e.\ $\sum_{i=1}^n x_i$).

Note that $\rho_0 = \Opt(\gamma_0)$, and therefore we may use relation $\rho_0$
as well (by Lemma~\ref{lmSurjectiveClosureReduction}). We construct a $\VCSPs$
instance $I$ as follows: Let $x_1, \dots, x_n$
be variables corresponding to the elements of the sought vector $\tup{x}$. The
requirement $H \cdot \tup{x} = \tup{0}^m$ can be seen as a system of $m$ linear
equations, each in the form $\bigoplus_{i=1}^k x_{a_i} = 0$ for a set $\{a_1,
\dots, a_k\} \subseteq [n]$ (the set may differ for each equation). We encode such an equation by introducing new
variables $y_0, \dots, y_k$ and imposing constraints $\rho_0(y_0)$,
$A_3(y_{i-1}, x_{a_i}, y_i)$ for all $i \in [k]$, and $\rho_0(y_k)$. These
ensure that each variable $y_j$ is assigned the value of the prefix sum
$\bigoplus_{i=1}^j x_{a_i}$, and that the total sum equals $0$. Finally, we
encode the objective function of the Minimum Distance problem by imposing
constraints $\gamma_0(x_1), \dots, \gamma_0(x_n)$.

Every vector $\tup{x} \in D^n$ satisfying $H \cdot \tup{x} = \tup{0}^m$
corresponds to a feasible assignment to $I$. If $\tup{x}$ is non-zero, the
corresponding assignment is surjective, as at least one of variables $x_1,
\dots, x_n$ gets label $1$ and, for every equation, variable $y_0$ gets label
$0$. Conversely, if a feasible assignment to $I$ is surjective, then it
corresponds to a non-zero vector $\tup{x}$ (labelling all variables $x_1, \dots,
x_n$ with $0$ implies that all the prefix sums $y_j$ equal $0$ as well). The
objective value of the assignment corresponding to a vector $\tup{x}$ equals the
weight of $\tup{x}$, and hence finding an optimal surjective solution to $I$
solves the Minimum Distance problem.

Finally, we show that $\{A_4, \gamma_=\}$ is s-intractable by a reduction from
$\VCSPs(\{A_3, \gamma_0\})$ to $\VCSPs(\{A_4, \gamma_=\})$. Given an instance
$I$, we construct an instance $I'$ by introducing a new variable $w$ and
replacing each constraint of the form $A_3(x,y,z)$ with $A_4(x,y,z,w)$ and each
constraint of the form $\gamma_0(x)$ with $\gamma_=(x,w)$. Any surjective
assignment to $I$ can be extended to a surjective assignment to $I'$ of the same
objective value by labelling $w$ with $0$. Conversely, consider a feasible
surjective assignment $s'$ to $I'$; we may assume $s'(w) = 0$ since language
$\{A_4, \gamma_=\}$ admits multimorphism $\mmorp{\neg}$. Restricting $s'$ to the
variables of $I$ gives us a surjective assignment to $I$ of the same objective
value. Note that if $s'$ assigns label $1$ to all the variables except $w$, its
restriction will not be surjective; however, such $s'$ violates constraints
$\rho_0(y_0)$ and thus is not feasible.\qedhere

\end{proof}

\begin{lemma}
\label{lmNotEDSsIntractable}

Let $\Gamma$ be a language such that it admits multimorphism $\mmorp{c_0}$ but
not polymorphism $\sub$. If $\Gamma \cup \mathcal{C}_D$ is intractable, then
$\VCSPs(\Gamma)$ is NP-hard.

\end{lemma}

\begin{proof}

Let $\Phi = \Feas(\Gamma) \cup \Opt(\Gamma) \subseteq \clGamma$. Suppose that
$\Phi$ does not admit any of the following polymorphisms: $\min$, $\max$,
$\mnrt$, and $\mjrt$. By the classification of $\{0, \infty\}$-valued languages
(see Theorem~\ref{thmClassificationSurjectiveCSP}), $\Phi$ is s-intractable. Hence,
$\VCSPs(\Gamma)$ is NP-hard by Lemma~\ref{lmSurjectiveClosureReduction}. In the rest
of the proof, we assume that $\Phi$ admits at least one of polymorphisms $\min$,
$\max$, $\mnrt$, and $\mjrt$. Note that $\Phi$ admits polymorphism $c_0$ but not
polymorphism $\sub$. Since $\min(x, y) = \mjrt(x, y, 0)$, we may assume that
$\Phi$ admits at least one of polymorphisms $\min$, $\max$, and $\mnrt$.

Suppose that $\Phi$ admits polymorphism $\neg$. Then it does not admit $\min$,
as $\sub(x, y) = \min(x, \neg y)$, nor it admits $\max$, as $\min(x, y) = \neg
\max(\neg x, \neg y)$. Therefore, $\Phi$ admits polymorphism $\mnrt$. If
$\Gamma$ is crisp, then language $\Gamma \cup \{\rho_0, \rho_1\}$ admits
multimorphism $\mmorp{\mnrt, \mnrt, \mnrt}$ and thus is tractable by
Theorem~\ref{thmClassificationVCSP}, which contradicts an assumption of the lemma.
Hence, $\Gamma$ is not crisp. By Lemmas~\ref{lmNegMnrtNotEDS}~and~\ref{lmNonCrisp}, we have
$\{A_4, \gamma_=\} \subseteq \clGamma$. Therefore, $\VCSPs(\Gamma)$ is NP-hard
by Lemma~\ref{lmA34sIntractable}.

If $\Phi$ does not admit polymorphism $\neg$, then, by Lemma~\ref{lmNotNeg}, we have
$\rho_0 \in \clGamma$ or $\rho_\leq \in \clGamma$. If $\rho_\leq \in \clGamma$,
$\VCSPs(\Gamma)$ is NP-hard by Lemma~\ref{lmReductionLeq} and we are done; in the
rest of the proof we assume that $\rho_\leq \not\in \clGamma$ and hence $\rho_0
\in \clGamma$. If $\Phi$ admits polymorphism $\min$ or $\max$, we get $\rho_\leq
\in \clGamma$ by Lemma~\ref{lmNotNegEDS}, which is a contradiction. Therefore, $\Phi$
admits $\mnrt$, and thus $\Gamma$ is not crisp (by the same argument as in the
previous paragraph). By Lemmas~\ref{lmNotNegEDS}~and~\ref{lmNonCrisp}, we have $\{A_3,
\gamma_0\} \subseteq \clGamma$. Therefore, $\VCSPs(\Gamma)$ is NP-hard by
Lemma~\ref{lmA34sIntractable}. \qedhere

\end{proof}

\begin{theorem}
\label{thmSurjectiveHardness}

Let $\Gamma$ be a language such that it is not EDS, $\neg(\Gamma)$ is not EDS,
and $\Gamma$ does not admit any of the following multimorphisms: $\mmorp{\min,
\min}$, $\mmorp{\max, \max}$, $\mmorp{\min, \max}$, $\mmorp{\mnrt, \mnrt,
\mnrt}$, $\mmorp{\mjrt, \mjrt, \mjrt}$, $\mmorp{\mjrt, \mjrt, \mnrt}$. Then
$\VCSPs(\Gamma)$ is NP-hard.

\end{theorem}

\begin{proof}

If $\Gamma$ does not admit at least one of multimorphisms $\mmorp{c_0}$ and
$\mmorp{c_1}$, it is intractable by Theorem~\ref{thmClassificationVCSP}, and hence
$\VCSPs(\Gamma)$ is NP-hard by Lemma~\ref{lmReductionVCSPtoSurjective}. Language
$\Gamma \cup \mathcal{C}_D$ is, by the same theorem, intractable. We may assume
that $\Gamma$ admits multimorphism $\mmorp{c_0}$; if it does not, we consider
$\neg(\Gamma)$ instead. By Corollary~\ref{corNotEDSLanguageReduction} and
Lemma~\ref{lmNotEDSsIntractable}, $\VCSPs(\Gamma)$ is NP-hard. \qedhere

\end{proof}

\section{Tractability of EDS languages}
\label{secSurjectiveTractability}

We prove that EDS languages are globally s-tractable by a reduction to a
generalised variant of the Min-Cut problem. The problem is defined in
Section~\ref{ssecDefinitionGMC}, its tractability is established in
Section~\ref{ssecTractabilityGMC}, and the reduction is stated in
Section~\ref{ssecReductionToGMC}.

\subsection{Generalised Min-Cut problem}
\label{ssecDefinitionGMC}

Let $V$ be a finite set. A \emph{set function} on $V$ is a function $\gamma: 2^V
\to \QnInfty$ with $\gamma(\emptyset) = 0$.

\begin{definition}

A set function $\gamma : 2^V \to \QnInfty$ is \emph{symmetric} if $\gamma(X) =
\gamma(V \setminus X)$ for all $X \subseteq V$; it is \emph{increasing} if
$\gamma(X) \leq \gamma(Y)$ for all $X \subseteq Y \subseteq V$; it is
\emph{superadditive} if
\begin{equation}
\gamma(X) + \gamma(Y) \leq \gamma(X \cup Y)
\end{equation}
for all disjoint $X, Y \subseteq V$; it is \emph{posimodular} if
\begin{equation}
\gamma(X) + \gamma(Y) \geq \gamma(X \setminus Y) + \gamma(Y \setminus X)
\end{equation}
for all $X, Y \subseteq V$; and it is \emph{submodular} if
\begin{equation}
\gamma(X) + \gamma(Y) \geq \gamma(X \cap Y) + \gamma(X \cup Y)
\end{equation}
for all $X, Y \subseteq V$.

\end{definition}

Note that any superadditive set function is also increasing, as for all $X
\subseteq Y \subseteq V$ it holds $\gamma(X) \leq \gamma(X) + \gamma(Y \setminus
X) \leq \gamma(Y)$ by superadditivity. In the case of symmetric set functions,
submodularity implies posimodularity, as
\begin{align}
\gamma(X) + \gamma(Y)
&=
\gamma(X) + \gamma(V \setminus Y) \\
&\geq
\gamma(X \cap (V \setminus Y)) + \gamma(X \cup (V \setminus Y)) \\
&=
\gamma(X \setminus Y) + \gamma(V \setminus (Y \setminus X)) \\
&=
\gamma(X \setminus Y) + \gamma(Y \setminus X) \,.
\end{align}
and, similarly, posimodularity implies submodularity.

\begin{example}
\label{exSetFunctions}
Let $V$ be a finite set and $T \subseteq V$ a non-empty subset. We define a set
function $\gamma$ on $V$ by $\gamma(X) = 1$ if $T \subseteq X$ and $\gamma(X) =
0$ otherwise. Intuitively, this corresponds to a soft NAND constraint imposed on
variables $T$. The set function $\gamma$ is superadditive, and hence also
increasing.
\end{example}

We now formally define the Min-Cut problem.

\begin{definition}
\label{defMinCut}

An instance of the \emph{Min-Cut} (MC) problem is given by an undirected graph
$G = (V, E)$ with edge weights $w: E \to \QnInfty$. The objective function $g$
of the MC problem is a set function on $V$ defined by
\begin{equation}
g(X)
=
\sum_{|X\cap\{u,v\}|=1} w(u,v) \,.
\end{equation}
Function $g$ is a well-known example of a submodular function. Since it is
symmetric, it is also posimodular.

A \emph{solution} to the MC problem is a set $X$ such that $\emptyset \subsetneq
X \subsetneq V$. Note that a cut $(X, V \setminus X)$ corresponds to two
solutions, namely $X$ and $V \setminus X$. An \emph{optimal} solution is a
solution with the minimum objective value among all solutions. A \emph{minimal}
optimal solution is an optimal solution with no proper subset being an optimal
solution.

\end{definition}

Note that any two different minimal optimal solutions $X, Y$ must be disjoint,
otherwise $X \setminus Y$ or $Y \setminus X$ would be a smaller optimal solution
(by the posimodularity of $g$).

Although the definition allows infinite weight edges, those can be easily
eliminated by identifying their endpoints, and so we may assume that all edge
weights are finite. Edges with weight $0$ are conventionally disregarded.

Finally, we define the Generalised Min-Cut problem, which further generalises
the problem introduced in \cite{Uppman12:cp}.

\begin{definition}
\label{defGMC}

An instance $\gmcinst$ of the \emph{Generalised Min-Cut} (GMC) problem is given
by an undirected graph $G = (V, E)$ with edge weights $w: E \to \QnInfty$, and
an oracle defining a superadditive set function $f$ on $V$. The objective
function the GMC problem is a set function on $V$ defined by $\gmcinst(X) = f(X)
+ g(X)$, where $g$ is the objective function of the underlying Min-Cut problem
on $G$.

A \emph{solution} to the GMC problem is a set $X$ such that $\emptyset
\subsetneq X \subsetneq V$. An \emph{optimal} solution is a solution with the
minimum objective value among all solutions. We denote this minimum objective
value by $\lambda$. For any $\alpha \geq 1$, an \emph{$\alpha$-optimal} solution
is a solution $X$ such that $\gmcinst(X) \leq \alpha \lambda$.

\end{definition}

We show in Theorem~\ref{thmTractabilityGMC} that, in the case of $0 < \lambda < \infty$
and a fixed $\alpha \geq 1$, there are only polynomially many $\alpha$-optimal
solutions and they can be found in polynomial time.

\subsection{Tractability of the Generalised Min-Cut problem}
\label{ssecTractabilityGMC}

In this section, we present a polynomial-time algorithm that solves the
Generalised Min-Cut problem. We assume that $w(u,v) \in \Qp$ for all edges
$(u,v)$.

\begin{lemma}
\label{lmLambdaZeroInfty}
There is a polynomial-time algorithm that, given an instance $\gmcinst$ of the
GMC problem, either finds a solution $X$ with $\gmcinst(X) = \lambda = 0$, or
determines that $\lambda = \infty$, or determines that $0 < \lambda < \infty$.
\end{lemma}

\begin{proof}
A solution $X$ with $\gmcinst(X) = f(X) + g(X) = 0$ satisfies $f(X) = g(X) = 0$,
and hence it does not cut any edge. Since the set function $f$ is increasing,
we may assume that $X$ is a single connected component. The algorithm simply
tries each connected component as a solution, which takes a linear number of
queries to the oracle for $f$.

The case of $\lambda = \infty$ occurs only if $f(X) = \infty$ for all solutions
$X$. Since $f$ is increasing, it is sufficient to check all solutions of size
$1$.
\end{proof}

In view of Lemma~\ref{lmLambdaZeroInfty}, we can assume that $0<\lambda<\infty$. Our
goal is to show that, for a given $\alpha\geq 1$, all $\alpha$-optimal solutions
to a GMC instance can be found in polynomial time. This is proved in
Theorem~\ref{thmTractabilityGMC}; before that we need to prove several auxiliary lemmas
on properties of the MC and GMC problems.

\begin{lemma}
\label{lmSubsetReduction}

For any instance $\gmcinst$ of the GMC problem on a graph $G=(V,E)$ and any non-empty
set $V' \subseteq V$, there is an instance $\gmcinst'$ on the induced subgraph $G[V']$
that preserves the objective value of all solutions $X \subsetneq V'$. In
particular, any $\alpha$-optimal solution $X$ of $\gmcinst$ such that $X \subsetneq V'$
is $\alpha$-optimal for $\gmcinst'$ as well.

\end{lemma}

\begin{proof}

Edges with exactly one endpoint in $V'$ need to be taken into account separately
because they do not appear in the induced subgraph. We accomplish that by
defining the new set function $f'$ by
\begin{equation}
f'(X)
=
f(X) + \sum_{u \in X} \sum_{v \in V \setminus V'} w(u,v)
\end{equation}
for all $X \subseteq V'$. By the construction, $f'$ is superadditive, and the
objective value $\gmcinst'(X)$ for any $X \subsetneq V'$ equals $\gmcinst(X)$.

Note that the minimum objective value for $\gmcinst'$ is greater than or equal to the
minimum objective value for $\gmcinst$. Therefore, any solution $X \subsetneq V'$ that
is $\alpha$-optimal for $\gmcinst$ is also $\alpha$-optimal for $\gmcinst'$.
\qedhere

\end{proof}

\begin{lemma}
\label{lmXoptimalYminimal}
Let $X$ be an optimal solution to an instance of the GMC problem over vertices
$V$ with $\lambda < \infty$, and $Y$ a minimal optimal solution to the
underlying MC problem. Then $X \subseteq Y$, $X \subseteq V\setminus Y$, or $X$
is an optimal solution to the underlying MC problem.
\end{lemma}

\begin{proof}
Assume that $X \not\subseteq Y$ and $X \not\subseteq V\setminus Y$. If $Y
\subseteq X$, we have $f(Y) \leq f(X)$ as $f$ is increasing, and hence
$f(Y)+g(Y) \leq f(X)+g(X) < \infty$. Therefore, $Y$ is optimal for the GMC
problem and $X$ is optimal for the MC problem. In the rest, we assume that $Y
\not\subseteq X$.

By the posimodularity of $g$ we have $g(X) + g(Y) \geq g(X \setminus Y) + g(Y
\setminus X)$. Since $Y \setminus X$ is a proper non-empty subset of $Y$, it
holds $g(Y \setminus X) > g(Y)$, and hence $g(X) > g(X \setminus Y)$. But then
$f(X)+g(X) > f(X \setminus Y)+g(X \setminus Y)$ as $\infty > f(X) \geq f(X
\setminus Y)$. Set $X \setminus Y$ is non-empty, and therefore contradicts the
optimality of $X$.
\end{proof}

The following lemma relates the number of optimal solutions and the number of
minimal optimal solutions to the MC problem. Note that this bound is tight for
(unweighted) paths and cycles with at most one path attached to each vertex.

\begin{lemma}
\label{lmMinimalOptimalBound}
For any instance of the MC problem on a connected graph with $n \geq 2$ vertices and
$p$ minimal optimal solutions, there are at most $p(p-1)+2(n-p)$ optimal
solutions.
\end{lemma}

We prove the lemma by induction on $n$, closely following the proof that
establishes the cactus representation of minimum cuts
in~\cite{frank2011connections}. We note that the cactus representation could be
applied directly to obtain a weaker bound of $p(p-1)+O(n)$ but we do not know
how to achieve the exact bound using it.

\begin{proof}

For $n=2$, the lemma holds as there are exactly two solutions and both are
minimal optimal. Assume $n\geq 3$. We denote the number of optimal solutions by
$s$. A solution $X$ is called a \emph{star} if $|X|=1$ or $|X|=n-1$, otherwise
it is called \emph{proper}.

First we consider the case where every optimal solution is a star. Let us denote
the minimum cuts by $(\{v_1\}, V\setminus\{v_1\}), \dots, (\{v_h\},
V\setminus\{v_h\})$. If $h=1$, then we have $s=p=2$ and the bound holds.
Otherwise, there are $2h$ optimal solutions but only $h$ of them are minimal
(i.e., $\{v_1\}, \dots, \{v_h\}$). Hence,
\begin{align}
p(p-1)+2(n-p)
&=
2h + (h-1)\cdot(h-2) -2 + 2(n-h) \\
&\geq
2h = s
\end{align}
as it holds $n \geq h \geq 2$ and $n\geq 3$.

From now on we assume that there is a proper optimal solution, and hence $n\geq
4$. We say that solutions $X, Y$ \emph{cross} if none of $X\setminus Y$,
$Y\setminus X$, $X\cap Y$, $V\setminus (X\cup Y)$ is empty. Note that only
proper solutions might cross. If every proper optimal solution is crossed by
some optimal solution, then the graph is a cycle with edges of equal weight
\cite[Lemma 7.1.3]{frank2011connections}.
In that case, there are $n(n-1)$ optimal solutions (all sets of contiguous
vertices except for $\emptyset$ and $V$) and $n$ minimal optimal solutions (all
singletons), and therefore the bound holds.

Finally, assume that there is a proper optimal solution that is not crossed by
any optimal solution, and denote the corresponding minimum cut by $(V_1, V_2)$.
For any optimal solution $X$, it must hold either $X \subseteq V_1$, $V_1
\subseteq X$, $X \subseteq V_2$, or $V_2 \subseteq X$. For $i \in \{1, 2\}$, let
$G_i$ be the result of shrinking $V_i$ into a new vertex $t_i$ so that the
weight of any edge $(t_i, v)$ for $v \in V\setminus V_i$ equals the sum of
weights of edges $(u, v)$ for $u \in V_i$. Denote by $n_i$, $p_i$, and $s_i$ the
number of vertices, minimal optimal solutions, and optimal solutions to $G_i$.
It holds $n = n_1 + n_2 - 2$. Consider any solution $X'$ of $G_i$: If $t_i
\not\in X'$, it corresponds to a solution $X = X'$ of the original graph $G$;
otherwise it corresponds to $X = X' \setminus \{t_i\} \cup V_i$. In both cases,
the objective values of $X'$ and $X$ in their respective problem instances are
equal. Therefore, any optimal solution $X$ of $G$ such that $X \subseteq V_2$ or
$V_1 \subseteq X$ corresponds to an optimal solution to $G_1$, and any optimal
solution to $G$ such that $X \subseteq V_1$ or $V_2 \subseteq X$ corresponds to
an optimal solution in $G_2$. Hence, $p = p_1 + p_2 - 2$ and $s = s_1 + s_2 -
2$, as only solutions $V_1$ and $V_2$ satisfy both conditions simultaneously. By
the inductive hypothesis, we get
\begin{align}
p(p-1)+2(n-p)
&= p_1(p_1-1)+2(n_1-p_1) + p_2(p_2-1)+2(n_2-p_2) \nonumber \\
&\phantom{=} +2(p_1-2)\cdot(p_2-2) - 2 \\
&\geq s_1+s_2-2 + 2(p_1-2)\cdot(p_2-2) \\
&\geq s
\end{align}
as it holds $p_1, p_2 \geq 2$.
\end{proof}

\begin{lemma}
\label{lmOptimalBound}

For any instance of the GMC problem on $n$ vertices with $0 < \lambda < \infty$,
the number of optimal solutions is at most $n(n-1)$. There is an algorithm that
finds all of them in polynomial time.

\end{lemma}

Note that the bound of $n(n-1)$ optimal solutions precisely matches the known
upper bound of $\binom{n}{2}$ for the number of minimum
cuts~\cite{karger1993global}; the bound is tight for cycles.

\begin{proof}
Let $t(n)$ denote the maximum number of optimal solutions for such instances on
$n$ vertices. We prove the bound by induction on $n$. If $n=1$, there are no
solutions and hence $t(1)=0$. For $n\geq 2$, let $Y_1, \dots, Y_p$ be the
minimal optimal solutions to the underlying MC problem. As there exists at least
one minimum cut and the minimal optimal solutions are all disjoint, it holds $2
\leq p \leq n$.

First, suppose that $\bigcup Y_i = V$. By Lemma~\ref{lmXoptimalYminimal}, any optimal
solution to the GMC problem is either a proper subset of some $Y_i$ or an
optimal solution to the underlying MC problem. Restricting solutions to a proper
subset of $Y_i$ is, by Lemma~\ref{lmSubsetReduction}, equivalent to considering a GMC
problem instance on vertices $Y_i$, and hence the number of such optimal
solutions is bounded by $t(|Y_i|) \leq |Y_i| \cdot (|Y_i|-1)$. Since it holds
$\sum |Y_i| = n$ and $|Y_i| \geq 1$ for all $i$, the sum $\sum |Y_i| \cdot
(|Y_i|-1)$ is maximised when $p-1$ of the sets $Y_i$ are singletons and the size
of the remaining one equals $n-p+1$. If the graph is connected, then, by
Lemma~\ref{lmMinimalOptimalBound}, there are at most $p(p-1)+2(n-p)$ optimal
solutions to the underlying MC problem. Adding these upper bounds we get
\begin{align}
&\hspace{1em} p(p-1)+2(n-p) + \sum_{i=1}^{p} |Y_i|\cdot(|Y_i|-1) \\
&\leq p(p-1)+2(n-p) + (p-1)\cdot 1 \cdot 0 + (n-p+1)\cdot(n-p) \\
&= n(n-1) -2(p-2)\cdot(n-p) \\
&\leq n(n-1) \,.
\end{align}

If the graph is disconnected, the sets $Y_1, \dots, Y_p$ are precisely its
connected components. The optimal solutions to the underlying MC problem are
precisely unions of connected components (with the exception of $\emptyset$ and
$V$), which means that there can be exponentially many of them. However, only
the sets $Y_1, \dots, Y_p$ themselves can be optimal solutions to the GMC
problem: We have $0 < \lambda \leq f(Y_i) + g(Y_i) = f(Y_i)$. Since $f$ is
superadditive, it holds
\begin{equation}
f(Y_{i_1} \cup \dots \cup Y_{i_k})
\geq
f(Y_{i_1}) + \dots + f(Y_{i_k})
\geq
k\lambda
\end{equation}
for any distinct $i_1, \dots, i_k$, and hence no union of two or more connected
components can be an optimal solution to the GMC problem. This gives us an upper
bound of $p \leq p(p-1)+2(n-p)$, and the rest follows as in the previous case.

Finally, suppose that $\bigcup Y_i \neq V$, and hence the graph is connected.
Let $Z = V \setminus \bigcup Y_i$. By Lemma~\ref{lmXoptimalYminimal}, any
optimal solution to the GMC problem is a proper subset of some $Y_i$, a proper
subset of $Z$, set $Z$ itself, or an optimal solution to the underlying MC
problem. Similarly as before, we get an upper bound of
\begin{align}
&\hspace{1em} p(p-1)+2(n-p) + \sum_{i=1}^{p} |Y_i|\cdot(|Y_i|-1) +
  |Z|\cdot(|Z|-1) + 1 \\
&\leq p(p-1)+2(n-p) + p\cdot 1 \cdot 0 + (n-p)\cdot(n-p-1) + 1 \\
&= n(n-1) -2(p-1)\cdot(n-p) + 1 \\
&\leq n(n-1) \,.
\end{align}

Using a procedure generating all minimum cuts \cite{Vazirani1992}, it is
straightforward to turn the above proof into a recursive algorithm that finds
all optimal solutions in polynomial time.
\end{proof}

\begin{lemma}
\label{lmXalphaYminimal}
Let $\alpha, \beta \geq 1$. Let $X$ be an $\alpha$-optimal solution to an
instance $\gmcinst$ of the GMC problem over vertices $V$ with $0 < \lambda <
\infty$, and $Y$ an optimal solution to the underlying MC problem. If $g(Y) <
\lambda / \beta$, then
\begin{equation}
\label{eqXYsplit}
\gmcinst(X\setminus Y) + \gmcinst(X\cap Y) <
\left(\alpha + \frac{2}{\beta}\right) \lambda \,;
\end{equation}
if $g(Y) \geq \lambda / \beta$, then $X$ is an $\alpha\beta$-optimal solution to
the underlying MC problem.
\end{lemma}

\begin{proof}
If $g(Y) \geq \lambda / \beta$, it holds $g(X) \leq \gmcinst(X) \leq
\alpha\lambda \leq \alpha\beta \cdot g(Y)$, and hence $X$ is an
$\alpha\beta$-optimal solution to the underlying MC problem. In the rest we
assume that $g(Y) < \lambda / \beta$.

Since $g$ is posimodular, we have
\begin{align}
g(X) + g(Y) &\geq g(X\setminus Y) + g(Y\setminus X) \\
g(Y) + g(Y\setminus X) &\geq g(X\cap Y) + g(\emptyset) \,,
\end{align}
and hence
\begin{equation}
g(X) + 2g(Y) \geq g(X\setminus Y) + g(X\cap Y) \,.
\end{equation}
By superadditivity of $f$, it holds $f(X) \geq f(X\setminus Y) + f(X\cap Y)$.
The claim then follows from the fact that $f(X) + g(X) + 2g(Y) < (\alpha +
2/\beta)\lambda$.
\end{proof}

Finally, we prove that $\alpha$-optimal solutions to the GMC problem can be
found in polynomial time.

\begin{theorem}
\label{thmTractabilityGMC}

For any instance $\gmcinst$ of the GMC problem on $n$ vertices with $0 < \lambda
< \infty$ and $\alpha \in \mathbb{Z}_{\geq 1}$, the number of $\alpha$-optimal
solutions is at most $n^{20\alpha-15}$. There is an algorithm that finds all of
them in polynomial time.

\end{theorem}

Note that for a cycle on $n$ vertices, the number of $\alpha$-optimal solutions
to the MC problem is $\Theta(n^{2\alpha})$, and thus the exponent in our bound
is asymptotically tight in $\alpha$.

\begin{proof}

Let $\beta \in \mathbb{Z}_{\geq 3}$ be a parameter. Throughout the proof, we
relax the integrality restriction on $\alpha$ and require only that
$\alpha\beta$ is an integer. For $\alpha=1$, the claim follows from
Lemma~\ref{lmOptimalBound}, therefore we assume $\alpha \geq 1 + 1/\beta$ in the
rest of the proof.

Define a linear function $\ell$ by
\begin{equation}
\ell(x)
=
\frac{2(\beta+1)}{\beta-2} \cdot \left( \beta x - 3 \right) \,.
\end{equation}
We prove that the number of $\alpha$-optimal solutions is at most
$n^{\ell(\alpha)}$; taking $\beta=4$ then gives the claimed bound. Function
$\ell$ was chosen as a slowest-growing function satisfying the following
properties required in this proof: It holds $\ell(x)+\ell(y) \leq
\ell(x+y-3/\beta)$ for any $x, y$, and $\ell(x) \geq 2\beta x$ for any $x \geq
1+1/\beta$.

We prove the bound by induction on $n + \alpha\beta$. As it trivially holds for
$n \leq 2$, we assume $n \geq 3$ in the rest of the proof. Let $Y$ be an optimal
solution to the underlying MC problem with $k = |Y| \leq n/2$. If $g(Y) \geq
\lambda / \beta$ then, by Lemma~\ref{lmXalphaYminimal}, any $\alpha$-optimal solution
to the GMC problem is an $\alpha\beta$-optimal solution to the underlying MC
problem. Since $g(Y) \geq \lambda / \beta > 0$, the graph is connected, and
hence there are at most
\begin{equation}
2^{2\alpha\beta} \binom{n}{2\alpha\beta}
\leq
n^{2\alpha\beta}
\leq
n^{\ell(\alpha)}
\end{equation}
such solutions by~\cite{karger1993global}. (In detail,
\cite[Theorem~6.2]{karger1993global} shows that the number of
$\alpha\beta$-optimal cuts in an $n$-vertex graph is $2^{2\alpha\beta-1}
\binom{n}{2\alpha\beta}$, and every cut corresponds to two solutions.)

From now on we assume that $g(Y) < \lambda / \beta$, and hence inequality
\eqref{eqXYsplit} holds. Upper bounds in this case may be quite loose; in
particular, we use the following inequalities:
\begin{align}
(k/n)^{\ell(\alpha)} &\leq
(k/n)^{\ell(1+1/\beta)} = (k/n)^{2(\beta+1)} \leq (k/n)^8 \leq
(k/n)(1/2)^7 = k/128n \\
(1/n)^{2\beta} &\leq (1/n)^6 \leq (1/n)(1/3)^5 < 1/128n \,.
\end{align}

Consider any $\alpha$-optimal solution to the GMC problem $X$.

If $X \subsetneq Y$, then, by Lemma~\ref{lmSubsetReduction}, $X$ is an
$\alpha$-optimal solution to an instance on vertices $Y$. By the induction
hypothesis, there are at most $k^{\ell(\alpha)} \leq (k/128n) \cdot
n^{\ell(\alpha)}$ such solutions.

Similarly, if $X \subsetneq V\setminus Y$, then $X$ is an $\alpha$-optimal
solution to an instance on vertices $V\setminus Y$, and there are at most
\begin{equation}
(n-k)^{\ell(\alpha)}
=
(1-k/n)^{\ell(\alpha)} \cdot n^{\ell(\alpha)}
\leq
(1-k/n) \cdot n^{\ell(\alpha)}
\end{equation}
such solutions.

If $Y \subsetneq X$, then $X \setminus Y$ is an $(\alpha-1+2/\beta)$-optimal
solution on vertices $V \setminus Y$ by \eqref{eqXYsplit} and the fact that
$\gmcinst(X \cap Y) \geq \lambda$. Similarly, if $V \setminus Y \subsetneq X$,
then $X \cap Y$ is an $(\alpha-1+2/\beta)$-optimal solution on vertices $Y$. In
either case, we bound the number of such solutions depending on the value of
$\alpha$: For $\alpha < 2-2/\beta$, there are trivially none; for $\alpha =
2-2/\beta$, Lemma~\ref{lmOptimalBound} gives a bound of $n(n-1) \leq
n^{\ell(\alpha)-2\beta}$; and for $\alpha > 2-2/\beta$ we get an upper bound of
$n^{\ell(\alpha-1+2/\beta)} \leq n^{\ell(\alpha)-2\beta}$ by the induction
hypothesis. The number of solutions is thus at most $n^{\ell(\alpha)-2\beta}
\leq (1/128n) \cdot n^{\ell(\alpha)}$ for any $\alpha$.

Finally, we consider $X$ such that $\emptyset \subsetneq X \setminus Y
\subsetneq V \setminus Y$ and $\emptyset \subsetneq X \cap Y \subsetneq Y$,
i.e., $X \setminus Y$ and $X \cap Y$ are solutions on vertices $V \setminus Y$
and $Y$ respectively. Let $i$ be the integer for which
\begin{equation}
\left(1+\frac{i}{\beta}\right) \lambda
\leq \gmcinst(X \cap Y)
< \left(1+\frac{i+1}{\beta}\right) \lambda \,.
\end{equation}
Then, by \eqref{eqXYsplit}, it holds $\gmcinst(X \setminus Y) <
(\alpha-1-(i-2)/\beta) \lambda$. Therefore, $X \cap Y$ is a
$(1+(i+1)/\beta)$-optimal solution on vertices $Y$ and $X \setminus Y$ is an
$(\alpha-1-(i-2)/\beta)$-optimal solution on vertices $V \setminus Y$. Since
$0 \leq i \leq (\alpha-2)\beta+1$, we can bound the number of such solutions by
the induction hypothesis as at most
\begin{align}
k^{\ell\left(1+\frac{i+1}{\beta}\right)} \cdot
(n-k)^{\ell\left(\alpha-1-\frac{i-2}{\beta}\right)}
& \leq
\left(\frac{k}{n}\right)^{\ell\left(1+\frac{i+1}{\beta}\right)} \cdot
n^{\ell\left(1+\frac{i+1}{\beta}\right)+\ell\left(\alpha-1-\frac{i-2}{\beta}\right)}
\\
& \leq
\left(\frac{k}{n}\right)^{2(\beta+1)} \cdot
\frac{1}{2^i} \cdot
n^{\ell(\alpha)} \,,
\end{align}
which is at most $2 \cdot (k/128n) \cdot n^{\ell(\alpha)}$ in total for all $i$.

By adding up the bounds we get that the number of $\alpha$-optimal solutions
is at most $n^{\ell(\alpha)}$. A polynomial-time algorithm that finds the
$\alpha$-optimal solutions follows from the above proof using a procedure
generating all $\alpha\beta$-optimal cuts~\cite{Vazirani1992}.
\end{proof}

\begin{remark}
For our reduction from the $\VCSPs$ over EDS languages, we need to find all
$\alpha$-optimal solutions to the GMC problem. However, if one is only
interested in a single optimal solution, the presented algorithm can be easily
adapted to an even more general problem.

Let $f, g$ be set functions on $V$ given by an oracle such that
$f:2^V\to\QnInfty$ is increasing and $g:2^V\to\Qn$ satisfies the posimodularity
and submodularity inequalities for intersecting pairs of sets (i.e.\ sets $X, Y$
such that neither of $X\cap Y, X \setminus Y, Y \setminus X$ is empty). The
objective is to minimise the sum of $f$ and $g$.

The case when the optimum value $\lambda = \infty$ can be recognised by checking
all solutions of size $1$. Assuming $\lambda < \infty$, note that the proof of
Lemma~\ref{lmXoptimalYminimal} works even for this more general problem. Let $Y$
be a minimal optimal solution to $g$. It follows that there is an optimal
solution $X$ to $f+g$ such that $X \subseteq Y$, $X \subseteq V \setminus Y$, or
$X$ is itself a \emph{minimal} optimal solution to $g$ (as $f$ is increasing).
We can find all minimal optimal solutions to $g$ in polynomial
time~\cite[Theorem~10.11]{nagamochi2008algorithmic}. Restricting $f,g$ to a
subset of $V$ preserves the required properties, and hence we can recursively
solve the problem on $Y$ and $V \setminus Y$. Therefore, an optimal solution to
$f+g$ can be found in polynomial time.
\end{remark}

\subsection{Reduction to the Generalised Min-Cut problem}
\label{ssecReductionToGMC}

At the heart of our reduction is an observation that EDS weighted relations can
be approximated by instances of the Generalised Min-Cut problem. We define this
notion of approximability in Definition~\ref{defGMCApproximability}. In
Theorem~\ref{thmEDSSomeApproximation}, we show how to approximate any EDS weighted
relation with a constant factor. However, that construction does not yield a
sufficient bound on the approximation factor; we present it only in order to
provide some intuition for the more opaque construction in
Theorem~\ref{thmEDSBetterApproximation}. Using that, we establish the global
s-tractability of EDS languages in Theorem~\ref{thmTractabilityEDS}.

In this section, we equate weighted relations admitting multimorphism
$\mmorp{c_0}$ with set functions; the correspondence is formally stated in the
following definition. Note that we may without loss of generality assume that
the minimum assigned value equals $0$, as adding a rational constant to a
weighted relation preserves tractability.

\begin{definition}
Let $\gamma$ be an $r$-ary weighted relation such that, for any $r$-tuple
$\tup{x}$, $\gamma(\tup{x}) \geq \gamma(\tup{0}^r) = 0$. The \emph{corresponding
set function} $\gamma'$ on $[r]$ is defined by $\gamma'(X) = \gamma(\tup{x})$
where $x_i = 1 \iff i \in X$.
\end{definition}

The definition of $\alpha$-EDS weighted relations then translates into the
following:

\begin{definition}
\label{defSetFnEDS}

For any
$\alpha \geq 1$, a set function $\gamma$ on $V$ is \emph{$\alpha$-EDS}
if, for every $X, Y \subseteq V$, it holds
\begin{equation}
\alpha \cdot (\gamma(X) + \gamma(Y)) \geq \gamma(X \setminus Y) \,.
\label{eqDefEDSsetFn}
\end{equation}

\end{definition}

\begin{remark}

Inequality \eqref{eqDefEDSwRel} could be modified so that \eqref{eqDefEDSsetFn}
becomes symmetric, say
\begin{equation}
\alpha \cdot (\gamma(X) + \gamma(Y))
\geq
\gamma(X \setminus Y) + \gamma(Y \setminus X) \,.
\end{equation}
It is easy to see that, although the set of $\alpha$-EDS weighted relations for
a fixed $\alpha$ would be different, this change would not affect the set of EDS
weighted relations. We opt for the shorter, albeit asymmetric, definition.

\end{remark}

\begin{definition}
\label{defGMCApproximability}

Let $\gmcinst$ be an instance of the GMC problem on vertices $V$ and $\gamma$ a
set function on $V$. For any $\alpha \geq 1$, we say that $\gmcinst$
\emph{$\alpha$-approximates} $\gamma$ if, for all $X \subseteq V$,
\begin{equation}
\gmcinst(X) \leq \gamma(X) \leq \alpha \cdot \gmcinst(X) \,.
\end{equation}

A set function is \emph{$\alpha$-approximable} if there exists a GMC instance
that $\alpha$-approximates it, and it is \emph{approximable} if it is
$\alpha$-approximable for some $\alpha \geq 1$.
\end{definition}

\begin{theorem}
\label{thmEDSSomeApproximation}
Any $\alpha$-EDS set function is approximable.
\end{theorem}

\begin{proof}

Let $\gamma$ be an $\alpha$-EDS set function on $[n]$ and $\gamma'$ the
corresponding $n$-ary weighted relation. By Corollary~\ref{corFiniteEDSSub}, both
$\Feas(\gamma')$ and $\Opt(\gamma')$ are essentially downsets. The rest of the
proof relies only on this property and does not depend on the value of $\alpha$.
The intuition behind our construction is that a downset can be represented by a
superadditive function on $[n]$, and binary equality relations can be
represented by edges.

There exist $A_{\Feas}, A_{\Opt} \subseteq [n]$, downsets $\mathcal{S}_{\Feas}
\subseteq 2^{A_{\Feas}}$, $\mathcal{S}_{\Opt} \subseteq 2^{A_{\Opt}}$, and sets
of pairs of distinct coordinates $E_{\Feas}, E_{\Opt}$ such that $|A_{\Feas}| +
|E_{\Feas}| = |A_{\Opt}| + |E_{\Opt}| = n$ and
\begin{align}
\gamma(X) < \infty &\iff
X \cap A_{\Feas} \in \mathcal{S}_{\Feas} ~\land~
|X\cap\{i,j\}| \neq 1 \text{ for all } \{i, j\} \in E_{\Feas} \\
\gamma(X) = 0 &\iff
X \cap A_{\Opt} \in \mathcal{S}_{\Opt} ~\land~
|X\cap\{i,j\}| \neq 1 \text{ for all } \{i, j\} \in E_{\Opt} \,.
\end{align}

We construct an instance $\gmcinst$ of the GMC problem on vertices $[n]$ as
follows. Let $w_{\Feas}(i,j) = \infty$ if $\{i,j\} \in E_{\Feas}$ and
$w_{\Feas}(i,j) = 0$ otherwise. Let $w_{\Opt}(i,j) = 1$ if $\{i,j\} \in
E_{\Opt}$ and $w_{\Opt}(i,j) = 0$ otherwise. Then the weight of edge $(i,j)$ is
$w(i,j) = w_{\Feas}(i,j) + w_{\Opt}(i,j)$. Let $f_{\Feas}$ be a set function on
$[n]$ defined by $f_{\Feas}(X) = 0$ if $X \cap A_{\Feas} \in
\mathcal{S}_{\Feas}$ and $f_{\Feas}(X) = \infty$ otherwise; $f_{\Feas}$ is
superadditive because $\mathcal{S}_{\Feas}$ is a downset. Let $f_{\Opt}$ be a
set function on $[n]$ defined by $f_{\Opt}(X) = 0$ if $X \cap A_{\Opt} \in
\mathcal{S}_{\Opt}$ and $f_{\Opt}(X) = |X \cap A_{\Opt}|$ otherwise; $f_{\Opt}$
is superadditive because $\mathcal{S}_{\Opt}$ is a downset. Then the
superadditive function defining instance $\gmcinst$ is $f = f_{\Feas} +
f_{\Opt}$.

By the construction, it holds $\gamma(X) < \infty \iff \gmcinst(X) < \infty$ and
$\gamma(X) = 0 \iff \gmcinst(X) = 0$. Moreover, for any $X$ such that $0 <
\gmcinst(X) < \infty$, it holds $1 \leq \gmcinst(X) \leq n$. If the set
\begin{equation}
B
=
\left\{ \gamma(X) ~\middle|~
X \subseteq [n] ~\land~ 0 < \gamma(X) < \infty \right\}
\end{equation}
is empty, then instance $\gmcinst$ $1$-approximates $\gamma$; otherwise let
$b_\text{min}, b_\text{max}$ denote the minimum and maximum of $B$. We scale
the weights of the edges $w$ and the superadditive function $f$ by a factor of
$b_\text{min} / n$ to obtain an instance $\gmcinst'$ such that $\gmcinst'(X)
\leq \gamma(X)$ for all $X$. Instance $\gmcinst'$ then $(n \cdot b_\text{max} /
b_\text{min})$-approximates $\gamma$. \qedhere

\end{proof}

To establish the tractability of infinite EDS languages, we need a better bound
on the approximability of $\alpha$-EDS set functions than the one given in
Theorem~\ref{thmEDSSomeApproximation}. This is achieved in
Theorem~\ref{thmEDSBetterApproximation}, which we prove using the following
technical lemma. We refer the reader to~\cite[Theorem~1.1]{devanur2013approximation} for an
example of the application of this proof technique in a simpler setting.

\begin{lemma}
\label{lmEDSBetterApproximationCutLB}
Let $\gamma$ be an $\alpha$-EDS set function on $V$ for some $\alpha \geq 1$.
For any distinct $u, v \in V$, let $\Tuv$ be a subset of $V$ such that
$\left|\Tuv \cap \{u,v\}\right| = 1$. Then, for any $R \subseteq S \subseteq V$,
it holds
\begin{equation}
\alpha^{|S|+2} \cdot \left( \left(|S|^2+2\right) \cdot \gamma(S) +
\sum_{|R\cap\{u,v\}|=1} \gamma\left(\Tuv\right) \right) \geq \gamma(R) \,.
\end{equation}
\end{lemma}

\begin{proof}

First, we show by induction that, for any $X, Y_1, \dots, Y_n \subseteq V$, it
holds
\begin{equation}
\alpha^n \cdot \left(\gamma(X) + \sum_{i=1}^n \gamma(Y_i)\right) \geq
\gamma\left(X \setminus \bigcup_{i=1}^n Y_i\right) \,.
\label{eqLmCutLBGeneralizedAssumption}
\end{equation}
For $n = 1$, this is equivalent to \eqref{eqDefEDSsetFn}. As for the inductive
step, assume that \eqref{eqLmCutLBGeneralizedAssumption} holds for $n \geq 1$.
By the inductive hypothesis and \eqref{eqDefEDSsetFn}, we get
\begin{align}
\alpha^{n+1} \cdot \left( \gamma(X) + \sum_{i=1}^{n+1} \gamma(Y_i) \right)
&\geq
\alpha \cdot \left(
  \gamma \left( X \setminus \bigcup_{i=1}^n Y_i \right)
  +
  \gamma(Y_{n+1})
\right) \\
&\geq
\gamma \left( X \setminus \bigcup_{i=1}^{n+1} Y_i \right) \,.
\end{align}

If $\gamma(S) = \infty$, the inequality claimed by this lemma trivially holds.
In the rest of the proof, we assume $\gamma(S) < \infty$. For any $u \in R$, $v
\in S \setminus R$, we define a set $\oTuv$ such that $\oTuv \cap \{u,v\} =
\{v\}$: If $v \in \Tuv$, let $\oTuv = \Tuv$; otherwise let $\oTuv = S \setminus
\Tuv$. We claim that
\begin{equation}
\alpha \cdot \left( \gamma(S) + \gamma\left(\Tuv\right) \right)
\geq
\gamma(\oTuv) \,.
\label{eqOTuvUB}
\end{equation}
This is trivially true in the case of $\oTuv = \Tuv$, and it follows from
\eqref{eqDefEDSsetFn} in the case of $\oTuv = S \setminus \Tuv$. By
\eqref{eqOTuvUB}, it holds
\begin{align}
\sum_{|R\cap\{u,v\}|=1} \gamma\left(\Tuv\right)
&\geq
\sum_{u\in R} \sum_{v\in S \setminus R} \gamma\left(\Tuv\right) \\
&\geq
\frac{1}{\alpha} \sum_{u\in R} \sum_{v\in S \setminus R} \gamma(\oTuv)
 - |R| \cdot |S\setminus R| \cdot \gamma(S) \,.
\label{eqLmCutLB1}
\end{align}

For any $u \in R$, let
\begin{equation}
W_u = S \setminus \bigcup_{v \in S \setminus R} \oTuv \,.
\end{equation}
By properties of $\oTuv$, it holds $u \in W_u \subseteq R$. Moreover, we have
\begin{equation}
\alpha^{|S \setminus R|} \cdot \left( \gamma(S) + \sum_{v\in S \setminus R}
\gamma(\oTuv) \right) \geq \gamma(W_u)
\end{equation}
by \eqref{eqLmCutLBGeneralizedAssumption}, which together with
\eqref{eqLmCutLB1} gives us
\begin{align}
\sum_{|R\cap\{u,v\}|=1} \gamma\left(\Tuv\right)
&\geq
\frac{1}{\alpha^{|S \setminus R|+1}} \sum_{u\in R} \gamma(W_u)
 - |R| \cdot (|S\setminus R|+1) \cdot \gamma(S) \\
&\geq
\frac{1}{\alpha^{|S \setminus R|+1}} \sum_{u\in R} \gamma(W_u)
 - |S|^2 \cdot \gamma(S) \,.
\label{eqLmCutLB2}
\end{align}

As it holds $\bigcup_{u\in R} W_u = R$, we have
\begin{equation}
\alpha^{|R|} \cdot \left( \gamma(S) + \sum_{u\in R} \gamma(W_u) \right) \geq
\gamma(S \setminus R) \,,
\end{equation}
and hence
\begin{equation}
\sum_{|R\cap\{u,v\}|=1} \gamma\left(\Tuv\right)
\geq
\frac{1}{\alpha^{|S|+1}} \cdot \gamma(S \setminus R)
 - \left(|S|^2+1\right) \cdot \gamma(S) \,.
\label{eqLmCutLB3}
\end{equation}
As it holds $\alpha \cdot (\gamma(S) + \gamma(S \setminus R)) \geq \gamma(R)$,
this proves the claimed inequality.
\end{proof}

\begin{theorem}
\label{thmEDSBetterApproximation}
Any $\alpha$-EDS set function on $V$ is
$\alpha^{n+2}\left(n^3+2n\right)$-approximable, where $n = |V|$.
\end{theorem}

\begin{proof}

Let $\gamma$ be an $\alpha$-EDS set function on $V$ for some $\alpha \geq 1$. We
construct an instance $\gmcinst$ of the GMC problem on vertices $V$ such that it
$\alpha^{n+2}\left(n^3+2n\right)$-approximates $\gamma$. The weight of edge $(u,
v)$ is
\begin{equation}
w(u, v) = \frac{1}{n^3+2n} \cdot
\min \left\{ \gamma(Z) ~\middle|~
Z \subseteq V \land |Z \cap \{u, v\}| = 1 \right\} \,.
\label{eqThmEDSBetterApproximationDefW}
\end{equation}
Let $f$ be a set function on $V$ defined as
\begin{equation}
f(X) = \frac{|X|}{n^3+2n} \cdot
\min \left\{ \left(|Z|^2+2\right) \cdot \gamma(Z) ~\middle|~
X \subseteq Z \subseteq V \right\} \,.
\label{eqThmEDSBetterApproximationDefF}
\end{equation}
We claim that $f$ is a superadditive set function. As $\gamma(\emptyset) = 0$,
it holds $f(\emptyset) = 0$. Consider any disjoint $X, Y \subseteq V$ and let $Z
\supseteq X \cup Y$ be a minimiser in \eqref{eqThmEDSBetterApproximationDefF} for $f(X
\cup Y)$. It holds $f(X) \leq |X| \cdot \left(|Z|^2+2\right) \cdot \gamma(Z) /
\left(n^3+2n\right)$ and $f(Y) \leq |Y| \cdot \left(|Z|^2+2\right) \cdot
\gamma(Z) / \left(n^3+2n\right)$, and hence
\begin{equation}
f(X) + f(Y) \leq
\frac{|X \cup Y|}{n^3+2n} \cdot \left(|Z|^2+2\right) \cdot \gamma(Z) =
f(X \cup Y) \,.
\end{equation}
The edge weights $w$ and superadditive set function $f$ define the GMC instance
$\gmcinst$. Now we prove that it $\alpha^{n+2}\left(n^3+2n\right)$-approximates
$\gamma$.

First, we show that $\gmcinst(R) \leq \gamma(R)$ for all $R \subseteq V$. By
\eqref{eqThmEDSBetterApproximationDefF}, we have $f(R) \leq |R| \cdot
\left(|R|^2+2\right) \cdot \gamma(R) / \left(n^3+2n\right)$. For any edge $(u,
v)$ cut by $R$ (i.e.\ $|R \cap \{u, v\}|=1$), it holds $w(u, v) \leq \gamma(R) /
\left(n^3+2n\right)$ by \eqref{eqThmEDSBetterApproximationDefW}, and hence $g(R) \leq
|R| \cdot |V \setminus R| \cdot \gamma(R) / \left(n^3+2n\right)$. Together, this
gives
\begin{equation}
\gmcinst(R) = f(R) + g(R) \leq
\frac{|R| \cdot \left(|R|^2+|V\setminus R|+2\right)}{n^3+2n} \cdot \gamma(R) \leq
\gamma(R) \,.
\end{equation}

Second, we show that $\alpha^{n+2}\left(n^3+2n\right) \cdot \gmcinst(R) \geq
\gamma(R)$ for all $R \subseteq V$. For $R = \emptyset$, the inequality holds,
as $\gmcinst(\emptyset) = \gamma(\emptyset) = 0$. Otherwise, let $S \supseteq R$
be a minimiser in \eqref{eqThmEDSBetterApproximationDefF} for $f(R)$, and $\Tuv$ a
minimiser in \eqref{eqThmEDSBetterApproximationDefW} for any edge $(u, v)$. It holds
\begin{align}
\left(n^3+2n\right) \cdot f(R)
&= |R| \cdot \left(|S|^2+2\right) \cdot \gamma(S) \geq
   \left(|S|^2+2\right) \cdot \gamma(S) \\
\left(n^3+2n\right) \cdot g(R)
&= \sum_{|R\cap\{u,v\}=1|} \gamma\left(\Tuv\right) \,,
\end{align}
and therefore, by Lemma~\ref{lmEDSBetterApproximationCutLB},
$\alpha^{n+2}\left(n^3+2n\right) \cdot \gmcinst(R) \geq
\alpha^{|S|+2}\left(n^3+2n\right) \cdot \gmcinst(R) \geq \gamma(R)$.
\end{proof}

\begin{theorem}
\label{thmTractabilityEDS}

Any EDS language is globally s-tractable.

\end{theorem}

\begin{proof}
Let $\Gamma$ be an EDS language and $\alpha' \geq 1$ such that every weighted
relation in $\Gamma$ is $\alpha'$-EDS. Without loss of generality, we may assume
that $\gamma\left(\tup{0}^{\ar(\gamma)}\right) = 0$ for every $\gamma \in
\Gamma$, and hence identify weighted relations with their corresponding set
functions. Weighted relations in $\Gamma$ are of bounded arity and therefore,
by Theorem~\ref{thmEDSBetterApproximation}, there exists $\alpha$ such that
every $\gamma \in \Gamma$ is $\alpha$-approximable. We will denote by
$\gmcinst_\gamma$ a GMC instance that $\alpha$-approximates $\gamma$.

Given a $\VCSPs(\Gamma)$ instance $I$ with an objective function
\begin{equation}
\phi_I'(x_1, \dots, x_n)
=
\sum_{i=1}^q w_i \cdot \gamma_i(\tup{x}^i) \,,
\end{equation}
we denote by $\phi_I$ the corresponding set function and construct a GMC
instance $\gmcinst$ that $\alpha$-approximates $\phi_I$. For $i \in [q]$, we
relabel the vertices of $\gmcinst_{\gamma_i}$ to match the variables in the
scope $\tup{x}^i$ of the $i$th constraint (i.e., vertex $j$ is relabelled to
$x^i_j$) and identify vertices in case of repeated variables. As the constraint
is weighted by a non-negative factor $w_i$, we also scale the weights of the
edges of $\gmcinst_{\gamma_i}$ and the superadditive function by $w_i$. (Note
that non-negative scaling preserves superadditivity.) Instance $\gmcinst$ is
then obtained by adding up GMC instances $\gmcinst_{\gamma_i}$ for all $i \in
[q]$.

Let $\tup{x} \in D^n$ denote a surjective assignment minimising $\phi_I'$, $X
\subseteq [n]$ the corresponding set $\{ i \in [n] ~|~ x_i = 1 \}$, $Y \subseteq
[n]$ an optimal solution to $\gmcinst$, and $\lambda = \gmcinst(Y)$. Since
$\gmcinst$ $\alpha$-approximates $\phi_I$, it holds
\begin{equation}
\lambda \leq \gmcinst(X) \leq \phi_I(X) \leq \phi_I(Y) \leq \alpha
\cdot \gmcinst(Y) = \alpha\lambda \,,
\end{equation}
and hence $X$ is an $\alpha$-optimal solution to $\gmcinst$. By
Lemma~\ref{lmLambdaZeroInfty}, we can determine whether $\lambda = 0$, in which case
any optimal solution to $\gmcinst$ is also optimal for $\phi_I$; and whether
$\lambda = \infty$. If $0 < \lambda < \infty$, we find all $\alpha$-optimal
solutions by Theorem~\ref{thmTractabilityGMC}. \qedhere

\end{proof}

We now prove Theorem~\ref{thmEnumerationSurjectiveVCSP}.

\begin{proof}

We only need to prove the theorem in the case of an EDS language (whether
$\Gamma$ or $\neg(\Gamma)$, which is symmetric), as the remaining classes of
globally s-tractable languages include constants $\mathcal{C}_D$ and thus admit
a polynomial-delay algorithm using standard self-reduction
techniques~\cite{Creignou97,Cohen04:search}.

Let $\Gamma$ be an EDS language. As in the proof of Theorem~\ref{thmTractabilityEDS},
we may assume that every weighted relation in $\Gamma$ assigns $0$ as the
minimum value. Given an instance of $\VCSPs(\Gamma)$, we can determine in
polynomial time, by Lemma~\ref{lmLambdaZeroInfty}, whether $\lambda=0$,
$0<\lambda<\infty$, or $\lambda=\infty$. If $\lambda=0$, then optimal solutions
incur the minimum value from every constraint. By applying $\Opt$ to all
constraints, we obtain a CSP instance invariant under $\min$ (by
Lemma~\ref{lmSubImpliesC0Min}), and hence are able
to enumerate all optimal solutions with a polynomial delay by the results
in~\cite{Creignou97}. If $0<\lambda<\infty$, then the claim follows from the
proof of Theorem~\ref{thmTractabilityEDS}; moreover, the number of optimal solutions is
polynomially bounded (see Theorem~\ref{thmTractabilityGMC}). Finally, the case
$\lambda=\infty$ is trivial. \qedhere

\end{proof}

\section{Conclusions}

We have established the complexity classification of surjective VCSPs on
two-element domains. An obvious open problem is to consider surjective VCSPs on
three-element domains. A complexity classification is known for
$\{0,\infty\}$-valued languages~\cite{Bulatov06:3-elementJACM} and $\Q$-valued
languages~\cite{hkp14:sicomp} (the latter generalises the $\{0,1\}$-valued case
obtained in~\cite{Jonsson06:sicomp}). In fact, \cite{Kolmogorov17:sicomp}
implies a dichotomy for $\QInfty$-valued languages on a three-element domain.
However, all these results depend on the notion of core and the presence of
constants $\mathcal{C}_D$ in the language, and thus it is unclear how to use
them to obtain a complexity classification in the surjective setting. Moreover,
one special case of the CSP on a three-element domain is the
\emph{3-No-Rainbow-Colouring} problem~\cite{Bodirsky12:dam}, whose complexity
status is open.

\section*{Acknowledgements}

We would like to thank Yuni Iwamasa, who prompted us to extend the complexity
classification to languages of infinite size (and bounded arity). We also thank
the anonymous reviewers of the two extended
abstracts~\cite{Uppman12:cp,fz17:mfcs} of this paper.

\newcommand{\noopsort}[1]{}\newcommand{\Zivny}{\noopsort{ZZ}\v{Z}ivn\'y}


\begin{thebibliography}{10}

\bibitem{Bach11:surj}
Walter Bach and Hang Zhou.
\newblock Approximation for {M}aximum {S}urjective {C}onstraint {S}atisfaction
  {P}roblems.
\newblock Technical report, October 2011.
\newblock arXiv:1110.2953.

\bibitem{Barto14:jacm}
Libor Barto and Marcin Kozik.
\newblock Constraint {S}atisfaction {P}roblems {S}olvable by {L}ocal
  {C}onsistency {M}ethods.
\newblock {\em Journal of the ACM}, 61(1), 2014.
\newblock Article No. 3.

\bibitem{Barto17:survey}
Libor Barto, Andrei Krokhin, and Ross Willard.
\newblock Polymorphisms, and how to use them.
\newblock In Krokhin and \Zivny{} \cite{kz17:dagstuhl}, pages 1--44.

\bibitem{Bodirsky12:dam}
Manuel Bodirsky, Jan K{\'{a}}ra, and Barnaby Martin.
\newblock The complexity of surjective homomorphism problems - a survey.
\newblock {\em Discrete Applied Mathematics}, 160(12):1680--1690, 2012.

\bibitem{Bulatov06:3-elementJACM}
Andrei Bulatov.
\newblock A dichotomy theorem for constraint satisfaction problems on a
  3-element set.
\newblock {\em Journal of the ACM}, 53(1):66--120, 2006.

\bibitem{Bulatov17:focs}
Andrei Bulatov.
\newblock A dichotomy theorem for nonuniform {C}{S}{P}.
\newblock In {\em Proceedings of the 58th Annual IEEE Symposium on Foundations
  of Computer Science (FOCS'17)}, pages 319--330. IEEE, 2017.

\bibitem{Bulatov05:classifying}
Andrei Bulatov, Andrei Krokhin, and Peter Jeavons.
\newblock Classifying the {C}omplexity of {C}onstraints using {F}inite
  {A}lgebras.
\newblock {\em {SIAM} Journal on Computing}, 34(3):720--742, 2005.

\bibitem{Bulatov12:jcss-enumerating}
Andrei~A. Bulatov, V{\'{\i}}ctor Dalmau, Martin Grohe, and D{\'{a}}niel Marx.
\newblock Enumerating homomorphisms.
\newblock {\em Journal of Computer and System Sciences}, 78(2):638--650, 2012.

\bibitem{Chen2014}
Hubie Chen.
\newblock An algebraic hardness criterion for surjective constraint
  satisfaction.
\newblock {\em Algebra universalis}, 72(4):393--401, 2014.

\bibitem{Cohen04:search}
David~A. Cohen.
\newblock Tractable {D}ecision for a {C}onstraint {L}anguage {I}mplies
  {T}ractable {S}earch.
\newblock {\em Constraints}, 9:219--229, 2004.

\bibitem{cccjz13:sicomp}
David~A. Cohen, Martin~C. Cooper, P\'aid\'i Creed, Peter Jeavons, and Stanislav
  \Zivny.
\newblock An algebraic theory of complexity for discrete optimisation.
\newblock {\em SIAM Journal on Computing}, 42(5):915--1939, 2013.

\bibitem{cohen06:complexitysoft}
David~A. Cohen, Martin~C. Cooper, Peter~G. Jeavons, and Andrei~A. Krokhin.
\newblock The {C}omplexity of {S}oft {C}onstraint {S}atisfaction.
\newblock {\em Artificial Intelligence}, 170(11):983--1016, 2006.

\bibitem{Creignou95:jcss}
Nadia Creignou.
\newblock A dichotomy theorem for maximum generalized satisfiability problems.
\newblock {\em Journal of Computer and System Sciences}, 51(3):511--522, 1995.

\bibitem{Creignou97}
Nadia Creignou and Jean{-}Jacques H{\'{e}}brard.
\newblock On generating all solutions of generalized satisfiability problems.
\newblock {\em Informatique Th\'eorique et Applications}, 31(6):499--511, 1997.

\bibitem{crescenzi1997guide}
P.~Crescenzi.
\newblock A short guide to approximation preserving reductions.
\newblock In {\em Proceedings of the 12th Annual IEEE Conference on
  Computational Complexity}, CCC '97, pages 262--, Washington, DC, USA, 1997.
  IEEE Computer Society.

\bibitem{Dechter92:aaai}
Rina Dechter and Alon Itai.
\newblock Finding all solutions if you can find one.
\newblock In {\em AAAI 1992 Workshop on Tractable Reasoning}, pages 35--39,
  1992.

\bibitem{devanur2013approximation}
Nikhil~R. Devanur, Shaddin Dughmi, Roy Schwartz, Ankit Sharma, and Mohit Singh.
\newblock On the approximation of submodular functions.
\newblock April 2013.
\newblock arXiv:1304.4948.

\bibitem{Feder98:monotone}
Tom\'as Feder and Moshe~Y. Vardi.
\newblock The {C}omputational {S}tructure of {M}onotone {M}onadic {S{N}{P}} and
  {C}onstraint {S}atisfaction: {A} {S}tudy through {D}atalog and {G}roup
  {T}heory.
\newblock {\em {SIAM} Journal on Computing}, 28(1):57--104, 1998.

\bibitem{Fiala08:csr}
Ji\v{r}{\'{\i}} Fiala and Jan Kratochv{\'{\i}}l.
\newblock Locally constrained graph homomorphisms - structure, complexity, and
  applications.
\newblock {\em Computer Science Review}, 2(2):97--111, 2008.

\bibitem{Fiala05:tcs}
Ji\v{r}{\'{\i}} Fiala and Dani{\"{e}}l Paulusma.
\newblock A complete complexity classification of the role assignment problem.
\newblock {\em Theoretical Computer Science}, 349(1):67--81, 2005.

\bibitem{frank2011connections}
Andr\'as Frank.
\newblock {\em Connections in Combinatorial Optimization}.
\newblock Oxford Lecture Series in Mathematics and Its Applications. OUP
  Oxford, 2011.

\bibitem{fz16:toct}
Peter Fulla and Stanislav \Zivny.
\newblock {{A Galois Connection for Valued Constraint Languages of Infinite
  Size}}.
\newblock {\em ACM Transactions on Computation Theory}, 8(3), 2016.
\newblock Article No. 9.

\bibitem{fz17:mfcs}
Peter Fulla and Stanislav \Zivny.
\newblock The complexity of {B}oolean surjective general-valued {C}{S}{P}s.
\newblock In {\em Proceedings of the 42nd International Symposium on
  Mathematical Foundations of Computer Science (MFCS'17)}, 2017.

\bibitem{Golovach12:acta}
Petr~A. Golovach, Bernard Lidick{\'{y}}, Barnaby Martin, and Dani{\"{e}}l
  Paulusma.
\newblock Finding vertex-surjective graph homomorphisms.
\newblock {\em Acta Informatica}, 49(6):381--394, 2012.

\bibitem{Golovach12:tcs}
Petr~A. Golovach, Dani{\"{e}}l Paulusma, and Jian Song.
\newblock Computing vertex-surjective homomorphisms to partially reflexive
  trees.
\newblock {\em Theoretical Computer Science}, 457:86--100, 2012.

\bibitem{hkp14:sicomp}
Anna Huber, Andrei Krokhin, and Robert Powell.
\newblock Skew {b}isubmodularity and {v}alued {C}{S}{P}s.
\newblock {\em {SIAM} Journal on Computing}, 43(3):1064--1084, 2014.

\bibitem{johnson1988generating}
David~S. Johnson, Mihalis Yannakakis, and Christos~H. Papadimitriou.
\newblock On generating all maximal independent sets.
\newblock {\em Information Processing Letters}, 27(3):119 -- 123, 1988.

\bibitem{Jonsson06:sicomp}
Peter Jonsson, Mikael Klasson, and Andrei~A. Krokhin.
\newblock The approximability of three-valued {MAX} {CSP}.
\newblock {\em {SIAM} Journal on Computing}, 35(6):1329--1349, 2006.

\bibitem{karger1993global}
David~R. Karger.
\newblock Global {M}in-cuts in {RNC}, and {O}ther {R}amifications of a {S}imple
  {M}in-{C}ut {A}lgorithm.
\newblock In {\em Proceedings of the Fourth Annual ACM-SIAM Symposium on
  Discrete Algorithms (SODA'93)}, pages 21--30, 1993.

\bibitem{khot10:coco}
Subhash Khot.
\newblock On the unique games conjecture (invited survey).
\newblock In {\em Proceedings of the 25th Annual IEEE Conference on
  Computational Complexity (CCC'10)}, pages 99--121. IEEE Computer Society,
  2010.

\bibitem{Kolmogorov17:sicomp}
Vladimir Kolmogorov, Andrei~A. Krokhin, and Michal Rol{\'{\i}}nek.
\newblock {The Complexity of General-Valued CSPs}.
\newblock {\em {SIAM} Journal on Computing}, 46(3):1087--1110, 2017.

\bibitem{Kozik15:icalp}
Marcin Kozik and Joanna Ochremiak.
\newblock {Algebraic Properties of Valued Constraint Satisfaction Problem}.
\newblock In {\em Proceedings of the 42nd International Colloquium on Automata,
  Languages and Programming (ICALP'15)}, volume 9134 of {\em Lecture Notes in
  Computer Science}, pages 846--858. Springer, 2015.

\bibitem{kz17:dagstuhl}
Andrei~A. Krokhin and Stanislav \Zivny{}, editors.
\newblock {\em The Constraint Satisfaction Problem: Complexity and
  Approximability}, volume~7 of {\em Dagstuhl Follow-Ups}.
\newblock Schloss Dagstuhl - Leibniz-Zentrum fuer Informatik, 2017.

\bibitem{makarychev2017approximation}
Konstantin Makarychev and Yury Makarychev.
\newblock {Approximation Algorithms for CSPs}.
\newblock In Krokhin and \Zivny{} \cite{kz17:dagstuhl}, pages 287--325.

\bibitem{Martin15:jctb}
Barnaby Martin and Dani{\"{e}}l Paulusma.
\newblock The computational complexity of disconnected cut and
  2{K}$_{2}$-partition.
\newblock {\em Journal of Combinatorial Theory, Series {B}}, 111:17--37, 2015.

\bibitem{nagamochi2008algorithmic}
Hiroshi Nagamochi and Toshihide Ibaraki.
\newblock {\em Algorithmic aspects of graph connectivity}, volume 123.
\newblock Cambridge University Press New York, 2008.

\bibitem{raghavendra08:stoc}
Prasad Raghavendra.
\newblock Optimal algorithms and inapproximability results for every {C}{S}{P}?
\newblock In {\em Proceedings of the 40th Annual {A}{C}{M} Symposium on Theory
  of Computing (STOC'08)}, pages 245--254. ACM, 2008.

\bibitem{Raghavendra}
Prasad Raghavendra.
\newblock {Approximating {NP}-hard Problems: Efficient Algorithms and their
  Limits}.
\newblock {\em PhD Thesis}, 2009.

\bibitem{Schaefer78:complexity}
Thomas~J. Schaefer.
\newblock The {C}omplexity of {S}atisfiability {P}roblems.
\newblock In {\em Proceedings of the 10th {A}nnual {A}{C}{M} {S}ymposium on
  {T}heory of {C}omputing ({S}{T}{O}{C}'78)}, pages 216--226. ACM, 1978.

\bibitem{Schrijver03:CombOpt}
Alexander Schrijver.
\newblock {\em Combinatorial {O}ptimization: {P}olyhedra and {E}fficiency},
  volume~24 of {\em Algorithms and Combinatorics}.
\newblock Springer, 2003.

\bibitem{Soter97:jacm}
Mechthild Stoer and Frank Wagner.
\newblock A simple min-cut algorithm.
\newblock {\em Journal of the ACM}, 44(4):585--591, 1997.

\bibitem{tz16:jacm}
Johan Thapper and Stanislav \Zivny.
\newblock The complexity of finite-valued {C}{S}{P}s.
\newblock {\em Journal of the ACM}, 63(4), 2016.
\newblock Article No. 37.

\bibitem{Uppman12:cp}
Hannes Uppman.
\newblock Max-{S}ur-{C}{S}{P} on {T}wo {E}lements.
\newblock In {\em {P}roceedings of the 18th {I}nternational {C}onference on
  {P}rinciples and {P}ractice of {C}onstraint {P}rogramming ({C}{P}'12)},
  volume 7514 of {\em Lecture Notes in Computer Science}, pages 38--54.
  Springer, 2012.

\bibitem{Valiant79:sicomp-complexity}
Leslie~G. Valiant.
\newblock The {C}omplexity of {E}numeration and {R}eliability {P}roblems.
\newblock {\em {SIAM} Journal on Computing}, 8(3):410--421, 1979.

\bibitem{Vardy1997}
Alexander Vardy.
\newblock Algorithmic {C}omplexity in {C}oding {T}heory and the {M}inimum
  {D}istance {P}roblem.
\newblock In {\em Proceedings of the 29th Annual ACM Symposium on Theory of
  Computing (STOC'97)}, pages 92--109, New York, NY, USA, 1997. ACM.

\bibitem{Vazirani1992}
Vijay~V. Vazirani and Mihalis Yannakakis.
\newblock Suboptimal {C}uts: {T}heir {E}numeration, {W}eight and {N}umber
  ({E}xtended {A}bstract).
\newblock In {\em Proceedings of the 19th International Colloquium on Automata,
  Languages and Programming (ICALP'92)}, pages 366--377. Springer-Verlag, 1992.

\bibitem{Vikas13}
Narayan Vikas.
\newblock Algorithms for partition of some class of graphs under compaction and
  vertex-compaction.
\newblock {\em Algorithmica}, 67(2):180--206, 2013.

\bibitem{Zhuk17:focs}
Dmitriy Zhuk.
\newblock The {P}roof of {C}{S}{P} {D}ichotomy {C}onjecture.
\newblock In {\em Proceedings of the 58th Annual IEEE Symposium on Foundations
  of Computer Science (FOCS'17)}, pages 331--342. IEEE, 2017.

\end{thebibliography}
\end{document}